\documentclass[hidelinks]{lmcs} 
\usepackage[utf8]{inputenc}
\pdfoutput=1

\usepackage{silence}
\usepackage{lastpage}
\lmcsdoi{18}{3}{10}
\lmcsheading{}{\pageref{LastPage}}{}{}%
{Jul.~02,~2021}{Jul.~29,~2022}{}


\usepackage{amssymb}
\usepackage{hyperref}
\usepackage{stmaryrd}
\usepackage{pifont}
\usepackage{xcolor}
\usepackage{mathtools}
\usepackage{proof}
\usepackage{xypic}
\usepackage{booktabs}
\usepackage{amsthm}

\newcommand{\xmark}{\text{\ding{55}}}%

\newcommand{\cA}{\mathcal{A}}

\newcommand{\cC}{\mathcal{C}}

\newcommand{\cI}{\mathcal{I}}

\newcommand\cM[1][\Addrs]{\mathcal{M}_{#1}}
\newcommand\cT{\mathcal{T}}
\renewcommand\cR{\mathcal{R}}

\newcommand{\cS}{\mathcal{S}}

\newcommand{\und}{\ \land \ }
\newcommand{\imp}{\quad \Rightarrow \quad }

\newcommand{\eqbnf}{\,::=\,}
\newcommand{\st}{\mid}
\newcommand\set[1]{\{#1\}}
\newcommand{\rel}[1]{\text{\sc #1}} 
\newcommand{\nat}{\mathbb{N}}
\newcommand{\tuple}[1]{\langle #1\rangle}

\newcommand{\Lam}{\ensuremath{\Lambda}}
\newcommand{\Lamo}{\ensuremath{\Lambda^o}}
\newcommand{\lam}{\texorpdfstring{\ensuremath{\lambda}}{lambda}}
\newcommand{\Lama}{\ensuremath{\Lambda}(\Addrs)}
\newcommand{\lama}{\ensuremath{\lambda\Addrs}}
\newcommand{\comb}[1]{{\mathbf{#1}}}
\newcommand{\Var}{\mathrm{Var}}
\newcommand{\FV}[1]{\mathrm{FV}(#1)}
\newcommand{\subst}[2]{[#2/#1]}
\newcommand{\msto}[1][\beta]{\twoheadrightarrow_{#1}}
\newcommand{\invredd}[1][\beta]{{\,}_{#1}\!\!\twoheadleftarrow}

\newcommand{\be}{{\beta\eta}}
\newcommand{\Om}{\comb{\Omega}}
\newcommand{\cons}[1]{\hat #1}
\newcommand{\CL}{\mathrm{CL}}
\newcommand{\ssk}{\comb{k}}
\newcommand{\sss}{\comb{s}}
\newcommand{\ssi}{\comb{i}}
\newcommand{\sso}{\boldsymbol{\varepsilon}}
\newcommand{\am}{addressing machine}
\newcommand{\mach}{\mathsf} 

\newcommand\C{\mach C}
\newcommand\mM{\mach M}
\newcommand\mV{\mach V}
\newcommand\mZ{\mach Z}
\newcommand\mN{\mach N}
\newcommand\mK{\mach K}
\newcommand\mS{\mach S}
\newcommand{\Tapes}[1][\Addrs]{\mathbb{T}_{#1}}
\newcommand{\Addrs}{\mathbb{A}}
\newcommand{\Null}{\varnothing}
\newcommand{\Cons}[2]{#1::#2}
\newcommand{\appT}[2]{#1\,@\,#2\,}
\newcommand{\append}[2]{\appT{#1}{{[}#2{]}}} 
\newcommand{\repl}[2]{[#1:=#2]}
\newcommand{\stuck}[1]{\mathsf{stuck}(#1)}
\newcommand{\ins}[1]{\mathtt{#1}}
\newcommand{\RaS}[1]{\ins{Load}~#1}
\newcommand{\Load}[1]{\ins{Load}~#1}
 \newcommand{\Call}[1]{\ins{Call}~#1}
 \newcommand{\Apply}[3]{#3\shortleftarrow \ins{App}(#1,\,#2)}
\newcommand{\Lookinv}[1]{\#^{-1}(#1)}
\newcommand{\Lookup}{\#}
\newcommand{\App}[2]{#1\cdot #2}
\newcommand{\redrule}{(\ensuremath{\twoheadrightarrow_\Addrs})}
\newcommand{\redwerule}{(\ensuremath{\twoheadrightarrow_\Addrs^{\text{\ae}}})}
\newcommand{\extrule}{(\ensuremath{\text{\ae}})}
\newcommand{\goesto}{\Downarrow}
\newcommand{\red}[1][c]{\to_{\mach{#1}}}
\newcommand{\redd}[1][c]{\twoheadrightarrow_{\mach{#1}}}
\newcommand{\redh}{\to_{\mach{h}}}
\newcommand{\reddh}{\twoheadrightarrow_{\mach{h}}}
\newcommand{\equiva}{\equiv_\Addrs}
\newcommand{\equivea}{\equiv_\Addrs^{\text{\ae}}}
\newcommand{\sima}{\simeq_\Addrs}
\newcommand{\simea}{\simeq_\Addrs^{\text{\ae}}}
\newcommand{\eqea}{=_\Addrs^{\text{\ae}}}

\newcommand{\Val}[1]{\mathrm{Val}_{#1}}
\newcommand{\CInt}[2][\vec x]{\vert #2\vert_{#1}}
\newcommand{\Int}[2]{[\![ #1]\!]_{#2}}
\newcommand{\Th}[1]{\mathrm{Th}(#1)}
\newcommand{\blam}{\boldsymbol{\lambda}}
\newcommand{\blame}{\boldsymbol{\lambda\eta}}

\theoremstyle{plain} 

\def\bsub{\begin{enumerate}[(i)]}
\def\esub{\end{enumerate}}
\definecolor{red}{rgb}{1, 0, 0}

\begin{document}

\title[Instructions]{Addressing Machines as Models of \lam-Calculus}

\author[G.~Della Penna]{Giuseppe Della Penna\rsuper{a}}	
\address{Dep.\ of Information Engineering, Computer Science and Mathematics, University of L'Aquila, Italy}	
\email{giuseppe.dellapenna@univaq.it}  

\author[B.~Intrigila]{Benedetto Intrigila\rsuper{b}}	
\address{Dipartimento di Ingegneria dell'Impresa, University of Rome ``Tor Vergata'', Italy}	
\email{intrigil@mat.uniroma2.it}  

\author[G.~Manzonetto]{Giulio Manzonetto\lmcsorcid{0000-0003-1448-9014}\rsuper{c}}	
\address{Univ. Paris 13, Sorbonne Paris Cit\'e, LIPN, UMR 7030, CNRS, F-93430 Villetaneuse, France.}	
\email{manzonetto@univ-paris13.fr}  





\begin{abstract}
  \noindent Turing machines and register machines have been used for decades in theoretical computer science as abstract models of computation.
Also the \lam-calculus has played a central role in this domain as it allows to focus on the notion of functional computation, based on the substitution mechanism, while abstracting away from implementation details.
The present article starts from the observation that the equivalence between these formalisms is based on the Church-Turing Thesis rather than an actual encoding of \lam-terms into Turing (or register) machines. The reason is that these machines are not well-suited for modelling \lam-calculus programs.

We study a class of abstract machines that we call \emph{addressing machine} since they are only able to manipulate memory addresses of other machines. The operations performed by these machines are very elementary: load an address in a register, apply a machine to another one via their addresses, and call the address of another machine. We endow addressing machines with an operational semantics based on leftmost reduction and study their behaviour. The set of addresses of these machines can be easily turned into a combinatory algebra.
In order to obtain a model of the full untyped \lam-calculus, we need to introduce a rule that bares similarities with the $\omega$-rule and the rule $\zeta_\beta$ from combinatory logic.
\end{abstract}

\maketitle

\section*{Introduction}
In theoretical computer science several models of computation have been considered over the years, since the pioneering work of Turing~\cite{Turing36}.
Turing Machines (TMs) certainly played a crucial role in the understanding of the notion of computation, while Register Machines (RMs) are more adapted to represent programs executed in a von Neumann architecture~\cite{Rogers67}.
From a recursion-theoretic perspective, the class of partial recursive functions provides a natural description of those numeric functions that can be calculated by a mechanical device~\cite{Kleene36}.
In mathematical logic, \lam-calculus~\cite{Bare} and the related formalism --- combinatory logic~\cite{CurryF58} --- proved to be an inexhaustible source of inspiration for the development of formal systems, proof assistants and functional programming languages. As it is well-known, the basic computational mechanism of \lam-calculus is the {\em symbolic substitution of an expression for a variable}.
All these formalisms --- and many others that have been subsequently introduced --- are quite different, but they can be proved equivalent in the sense that they are capable of representing the same class of partial numerical functions, i.e.\ the class of partial recursive functions.
Despite the enormous importance of this result --- in particular as a strong evidence for the so called Turing-Church Thesis --- it is still of great interest to understand, at a deeper level, the relationships between the different computational formalisms.

In particular, the relationship between \lam-calculus and partial recursive functions was investigated by Henk Barendregt, who tried to build during his PhD a model of untyped \lam-calculus (\lam-model~\cite{Koymans82,Meyer82}) out of Kleene's partial combinatory algebra having the set of ``codes'' $\nat$ as underlying set and as application the partial operator $\{x\}(y)$, that can be interpreted as the possible result of applying the partial computable function with code $x$ to the input~$y$.
His intention was to use this binary operator $\{\cdot\}(\cdot)$ to construct a (total) combinatory algebra in such a way that Kleene's translation of \lam-calculus results would become a simple model-theoretic interpretation. It is important to observe that a direct approach cannot work, as recursive functions implicitly use the classic computational model, which requires that a function is \emph{strict} on its arguments, that is, the function is undefined whenever any of its arguments is undefined. On the other hand, both \lam-calculus and combinatory logic allow the representation of non-strict functions such as the combinator~$\comb{K}$.

Barendregt has set up several sophisticated constructions, but a definite solution is still missing.
The problem is nowadays receiving the attention of the scientific community because of the recent republication of his PhD thesis~\cite{BarendregtTh}, extended with commentaries.
On the bright side, these investigations led to the formulation of the famous $\omega$-rule because --- if such a \lam-model exists --- then it needs to satisfy this strong extensionality axiom.

Following the same line of research, but attacking the problem from a different angle, one might meaningfully wonder whether it is possible to construct a \lam-model based on appropriate abstract machines. The most obvious and canonical choice would be considering Turing Machines, but such an attempt has the same problem as the one encountered with recursive function, since TMs are strict on their arguments.
A second problem is how to represent higher-order computations: in an imperative programming language a function can take another function as argument by working with its address, but in a TM this would require to encode processes as data and then manipulate and execute such codes indirectly. This makes the simple, intuitive notion of communication through addresses  extremely difficult to realize.
To this day, no \lam-model of this kind has ever been constructed.

In this article we define a class of abstract machines, where the notions of  address and communication (through addresses) are not only crucial to model computation, but they become the unique ingredients available.
These machines are called \emph{addressing machines} and possess a finite tape from which they can read the input, some internal registers where they can store values read from the tape, and an internal program which is composed by a  list of instructions that are executed sequentially.
The input-tape and the internal registers are reminiscent of those in TMs and RMs, respectively.
Every machine is uniquely identified by its address, which is a value taken from a fixed countable set $\Addrs$. In this formalism, addresses are the only available data-type --- this means that both the input-tape and the internal registers of a machine (once initialized) contain addresses from $\Addrs$.
Programs are written in an assembly language possessing only three instructions\footnote{This choice is made on purpose, in the attempt of determining the minimum amount of operations giving rise to a Turing-complete formalism.}.
Besides reading its inputs, an \am{} can apply two addresses $a,b$ with each other and store the resulting address $a\cdot b$ in an internal register.
Intuitively, $a\cdot b$ is obtained by first taking the machine $\mach{M}$ having address $a$, then appending $b$ to its input-tape, and finally calculating the address of this new machine.
This application operation being static and manipulating addresses exclusively is \emph{total} even when the referenced machines are non-terminating once executed.
As a last step of its execution, an addressing machine can transfer the computation to another machine, possibly extending its input-tape, by retrieving its address from a register.
Although not crucial in the abstract definition of an addressing machine, it should be clear at this point that any implementation of this formalism requires the association between the machines and their addresses to be effective (see Section~\ref{sec:conclusions} for more details).

Addressing machines share with \lam-calculus the fact that there is no fundamental distinction between processes and data-types: in order to perform calculations on natural numbers a machine needs to manipulate the addresses of the corresponding numerals.
Another similarity is the fact that in both settings communication is achieved by transferring the computation from one entity to another one.
In the case of \am s, the machine currently ``in execution'' transfers the control by calling the address of another machine. In \lam-calculus, the subterm ``in charge'' is the one occupying the so-called ``head position'' and the control of the computation is transferred when the head variable is substituted by another term.
It is worth noting that process calculi such as the $\pi$-calculus also address communication using the concept of channel, where messages are exchanged~\cite{Milner99,Sangiorgi01}. This is not the kind of communication that we are going to model here: our form of communication is encoded in the notion of address, so that a machine receiving a message results in a new machine with a different address. In other words, we do not model the dynamics of the communication, but the evolution of the machine addresses actually encodes the effects of communication.
Another difference is the fact that $\pi$-calculus naturally models parallel computations as well as concurrency, while \am s are designed for representing sequential computations (one machine at a time is executed).

\paragraph{Contents.} The aim of the paper is twofold.
On the one side we want to present the class of \am s and analyze their fundamental properties. This is done in Section~\ref{sec:machines}, where we describe their  operational semantics in two different styles: as a term rewriting system (small-step semantics) and as a set of inference rules (big-step semantics). The two approaches are shown to be equivalent in case of addressing machines executing a terminating program (Proposition~\ref{prop:equivsem}).
On the other side, we wish to construct a model of the untyped \lam-calculus based on \am s, and study the interpretations of \lam-terms.
For this reason, we recall in the preliminary Section~\ref{sec:pre} the main facts about \lam-calculus, its equational theories and denotational models.
It turns out that the set $\Addrs$ of addresses, together with the operation of application previously described, is not a combinatory algebra (nor, \emph{a fortiori}, a \lam-model). In Section~\ref{sec:combalg} we show that it can be turned into a combinatory algebra by quotienting under an equivalence relation arising naturally from our small-step operational semantics. Two addresses are equivalent if the corresponding machines are interconvertible using a more liberal rewriting relation. From the confluence property enjoyed by this relation, we infer the consistency of the algebra (Proposition~\ref{prop:cAisnonextcombal}).
Unfortunately, the combinatory algebra so-obtained is not yet a model of \lam-calculus --- there are still $\beta$-convertible \lam-terms having different interpretations. 
Section~\ref{sec:consistency} is devoted to showing that a \lam-model actually arises when adding to the system a mild form of extensionality sharing similarities both with the $\omega$-rule in \lam-calculus~\cite{BarendregtTh} and with the rule $\zeta_\beta$ from combinatory logic~\cite{HindleyS86}. The consistency of the model follows from an analysis of the underlying ordinal.
Interestingly, the model itself is not extensional (Theorem~\ref{thm:Sinsonoextslm}).

\paragraph{Related works.} A preliminary version of \am s appeared in Della Penna's MSc thesis~\cite{DellaPennaTh}.
Other abstract machines having similar primitive instructions are present in the literature, but they were studied from the perspective of functional programs implementation,
see e.g.~\cite{FairbairnW87}. We do not claim that addressing machines are innovative, the originality of our work relies on the construction of a \lam-model (Section~\ref{sec:consistency}) and its analysis.
The practice of associating an address to a term is also well-established in the implementation of functional programming languages, and can be seen as the practical counterpart of explicit substitutions~\cite{LevyM99,BlancLM05,AccattoliCGC19}. The relationship between our \am s and explicit substitutions will be discussed  in Section~\ref{sec:conclusions}.

\section{Preliminaries}\label{sec:pre}

We present some notions that will be useful in the rest of the article.

\subsection{The Lambda Calculus --- Its Syntax}\label{subsec:lamcal}

For the \lam-calculus we mainly follow Barendregt's first book~\cite{Bare}.
We consider fixed a countable set $\Var$ of \emph{variables} denoted by $x,y,z,\dots$
\begin{defi}
The set $\Lam$ of \emph{\lam-terms} over $\Var$ is generated by the following simplified\footnote{
This basically means that parentheses are left implicit.} grammar (for $x\in\Var$):
\begin{equation}\tag{$\Lambda$}
M,N,P,Q \eqbnf\ x\mid \lam x.M\mid MN
\end{equation}
\end{defi}
We assume that application is left-associative and has a higher precedence than \lam-abstraction. Therefore $\lam x.\lam y.\lam z.xyz$ stands for $(\lam x.(\lam y.(\lam z.(xy)z)))$.
Moreover, we often write $\lam x_1\dots x_n.M$ for $\lam x_1\dots \lam x_n.M$.

\begin{defi} Let $M\in\Lam$.
\begin{enumerate}[(i)]
\item The set $\FV{M}$ of \emph{free variables} of $M$ is defined by induction:
\[
	\begin{array}{lll}
	\FV{x} &=& \set{x},\\
	\FV{\lam x.P} &=& \FV{P} -\set{x},\\
	\FV{PQ} &=&\FV{P}\cup\FV{Q}.\\
	\end{array}
\]
\item We say that $M$ is \emph{closed}, or \emph{a combinator}, whenever $\FV{M} = \emptyset$.
\item We let $\Lamo = \set{ M\in\Lam \st \FV{M} = \emptyset}$ be the set of all combinators.
\end{enumerate}
\end{defi}

\noindent
 The variables occurring in $M$ that are not free are called ``bound''.
From now on, \lam-terms are considered modulo \emph{$\alpha$-conversion}, namely, up to the renaming of bound variables (see~\cite[\S2.1]{Bare}).
\begin{nota} Concerning specific combinators we let:
\[
	\begin{array}{lcll}
	\comb{I}&=&\lam x.x,&\textrm{identity,}\\
	\comb{1}&=&\lam xy.xy,&\textrm{an $\eta$-expansion of the identity,}\\
	\comb{K}&=&\lam xy.x,&\textrm{first projection,}\\
	\comb{F}&=&\lam xy.y,&\textrm{second projection,}\\
	\comb{S}&=&\lam xyz.xz(yz),&\textrm{$S$-combinator from Combinatory Logic,}\\
	\comb{\Delta}&=&\lam x.xx,&\textrm{self-application,}\\
	\Om&=&\comb{\Delta\Delta},&\textrm{paradigmatic looping combinator,}\\
	\comb{Y}&=&\lam f.(\lam x.f(xx))(\lam x.f(xx)),&\textrm{Curry's fixed point combinator.}\\
	\end{array}
\]
\end{nota}
The \lam-calculus is given by the set $\Lam$ endowed with reduction relations that turn it into a higher-order term rewriting system.

We say that a relation $\rel{R}\subseteq\Lam^2$ is \emph{compatible} if it is compatible w.r.t.\ application and \lam-abstraction. This means that, for $M,N,P\in\Lam$, if $M \rel{\,R\,} N$ holds then also $MP \rel{\,R\,} NP$, $PM \rel{\,R\,}PN$ and $\lam x.M \rel{\,R\,} \lam x.N$ hold.

\begin{defi} Define the following reduction relations.
\bsub
\item The \emph{$\beta$-reduction} $\to_\beta$ is the least compatible relation closed under the rule
\begin{equation}\tag{$\beta$}
	(\lam x.M)N\to M\subst{x}{N}
\end{equation}
where $M\subst{x}{N}$ denotes the \lam-term obtained by substituting $N$ for all free occurrences of $x$ in $M$, subject to the usual proviso about renaming bound variables in $M$ to avoid capture of free variables in $N$.
\item Similarly, the \emph{$\eta$-reduction} $\to_\eta$ is the least compatible relation closed under the rule
\begin{equation}\tag{$\eta$}
	\lam x.Mx \to M, \textrm{ if }x\notin\FV{M}.
\end{equation}
\item Moreover, we define $\to_\be\ =\ \to_\beta\cup\to_\eta$.
\item
	The relations $\to_\beta,\to_\eta$ and $\to_\be$ respectively generate the notions of \emph{multi-step reduction}  $\msto,\msto[\eta],\msto[\beta\eta]$ (resp.\ \emph{conversion} $=_\beta,=_\eta,=_\be$) by taking the reflexive and transitive (and symmetric) closure.
\esub
\end{defi}

\begin{thm}[Church-Rosser]
The reduction relation $\msto[\beta(\eta)]$ is confluent:
\[
M\msto[\beta(\eta)] M_1 \und M\msto[\beta(\eta)]M_2\imp \exists N\in\Lam\,.\,
M_1\msto[\beta(\eta)] N \invredd[\beta(\eta)]M_2
\]
\end{thm}
The \lam-terms are classified into solvable and unsolvable, depending on their capability of interaction with the environment.
\begin{defi}
A \lam-term $M$ is called \emph{solvable} if $(\lam\vec x.M)\vec P =_\beta \comb{I}$ for some $\vec x$ and $\vec P\in\Lam$. Otherwise $M$ is called \emph{unsolvable}.
\end{defi}
We say that a \lam-term $M$ \emph{has a head normal form} (\emph{hnf}) if it reduces to a \lam-term of shape $\lambda x_{1}\ldots x_{n}.yM_{1}\cdots M_{k}$ for some $n,k\ge 0$.
As shown by Wadsworth in~\cite{Wadsworth76}, a \lam-term $M$ is solvable if and only if $M$ has a head normal form.
The typical examples of unsolvable \lam-terms are $\Om, \lam x.\Om$ and $\comb{YI}$.

\subsection{Lambda theories and lambda models}

Conservative extensions of $\beta$-conversion are known as ``\lam-theories'' and have been extensively studied in the literature, see e.g.~\cite{Bare,LusinS04,IntrigilaMP19,IntrigilaS17,ManzonettoPSS19}.

\begin{defi}\
\bsub
\item A \emph{\lam-theory} $\cT$ is any congruence on $\Lam^2$ including $\beta$-conversion $=_\beta$.
\item A \lam-theory $\cT$ is called:
\begin{itemize}
\item \emph{consistent}, if $\cT$ does not equate all \lam-terms;
\item \emph{inconsistent}, if $\cT$ is not consistent;
\item \emph{extensional}, if $\cT$ contains the $\eta$-conversion $=_\eta$ as well;
\item \emph{sensible}, if $\cT$ is consistent and equates all unsolvable \lam-terms;
\item \emph{semi-sensible}, if $\cT$ does not equate a solvable and an unsolvable.
\end{itemize}
\esub
We write $\cT\vdash M = N$, or simply $M =_\cT N$, whenever $(M,N)\in\cT$.
\end{defi}

The set of all \lam-theories, ordered by inclusion, forms a quite rich complete lattice.
We denote by $\blam$ (resp.\ $\blame$) the smallest (resp.\ extensional) \lam-theory.
Both $\blam$ and $\blame$ are consistent, semi-sensible but not sensible.
A \lam-theory can be introduced syntactically, or semantically as the theory of a model. The model theory of \lam-calculus is largely based on the notion of combinatory algebras, and its variations (see, e.g.,~\cite{Koymans82,Selinger02,Meyer82,HindleyLS72} and~\cite[Ch.~5]{Bare}).

\begin{defi}\
\bsub
\item An \emph{applicative structure} is given by $\cA = (A,\cdot\,)$ where $A$ is a set and $(\cdot)$ is a binary operation on $A$ called \emph{application}.
We represent application as juxtaposition and we assume it is left-associative, e.g., $abc = (a\cdot b)\cdot c$.
An equivalence $\simeq$ on $\cA$ is a \emph{congruence} if it is compatible w.r.t.\ application:
\[
	a \simeq a' \und b\simeq b'\imp ab\simeq a'b'
\]

\item
A \emph{combinatory algebra} $\cC = (C, \cdot, \ssk, \sss)$ is an applicative structure for a signature with two constants $\ssk,\sss$, such that $\ssk\neq\sss$ and  ($\forall x,y,z\in C$):
\[
\ssk xy=x, \textrm{ and }\sss xyz=xz(yz).
\]
We say that $\cC$ is \emph{extensional} if the following holds:
\[
	\forall x. \forall y.(\forall z .(xz = yz) \Rightarrow x = y)
\]
\item Given a combinatory algebra $\cC$ and a congruence $\simeq$ on $(C,\cdot\,)$, define:
\[
	\cC_\simeq = (C/_\simeq,\bullet_\simeq,\ssk_\simeq,\sss_\simeq)
\]
where
\[
	\begin{array}{rcl}
	{[}a{]}_\simeq\bullet_\simeq [b]_\simeq &=& [a\cdot b]_\simeq,\\ 
	\ssk_\simeq &=& [\ssk]_\simeq\\
	\sss_\simeq &=& [\sss]_\simeq.\\
	\end{array}
\]
It is easy to check that if $\ssk\not\simeq\sss$ then $\cC_\simeq$ is a combinatory algebra.
\esub
\end{defi}

\noindent
We call $\ssk$ and $\sss$ the \emph{basic combinators}; the derived combinators $\ssi$ and $\sso$ are defined by $\ssi= \sss\ssk\ssk$ and $\sso= \sss(\ssk\ssi)$.
It is not difficult to verify that every combinatory algebra satisfies the identities $\ssi x=x$ and $\sso xy=xy$.

It is well-known that combinatory algebras are models of combinatory logic. A \lam-term $M$ can be interpreted in any combinatory algebra $\cC$ by first translating $M$ into a term $X$ of combinatory logic, written $(M)_\CL = X$, and then interpreting the latter in $\cC$. However, there might be $\beta$-convertible \lam-terms $M,N$ that are interpreted as distinguished elements of~$\cC$. For this reason, not all combinatory algebras are actually models of \lam-calculus.

The axioms of an elementary subclass of combinatory algebras, called \emph{$\lambda$-models}, were expressly chosen to make coherent the definition of interpretation of $\lambda$-terms (see~\cite[Def.~5.2.1]{Bare}).
The \emph{Meyer-Scott axiom} is the most important axiom in the definition of a $\lambda$-model.
In the first-order language of combinatory algebras it becomes:
\[
\forall x. \forall y\,.\,(\forall z\,.\, (xz = yz) \Rightarrow \sso x = \sso y).
\]
The combinator $\sso$ becomes an inner choice operator, that makes coherent the interpretation of an abstraction $\lambda$-term.

\subsection{Syntactic \lam-models}

The definition of a \lam-model is difficult to handle in practice because the five Curry's axioms~\cite[Thm.~5.2.5]{Bare} are complicated to verify by hand.
To prove that a certain combinatory algebra is actually a \lam-model, it is preferable to exploit Hindley's (equivalent) notion of a syntactic \lam-model. See, e.g.,~\cite{Koymans82}.

The definition of syntactic \lam-model in~\cite{Koymans82} is general enough to interpret \lam-terms possibly containing constants $\cons{a}$ representing elements $a$ of a set $A$.
We follow that tradition and denote by $\Lam(A)$ the set of all \lam-terms possibly containing constants from $A$, and we call them \emph{$\lam A$-terms}. For instance, given $a\in A$, we have $M =\comb{I}(\lam x.x\cons a)\cons b\in\Lam(A)$.
All notions, notations and results from Subsection~\ref{subsec:lamcal} extend to $\lam A$-terms without any problem.
In particular, substitution is extended by setting $\cons a\subst{x}{N} = \cons a$, for all $a \in A$ and $N\in\Lam(A)$.
As an example, the $\lam A$-term $M$ above reduces as follows: $M \to_\beta (\lam x.x\cons a)\cons b\to_\beta \cons b\cons a\in\Lam(A)$.
Observe that substitutions of variables by constants always permute, namely $M\subst{x}{\cons a}\subst{y}{\cons b} = M\subst{y}{\cons b}\subst{x}{\cons a}$, for all $a,b\in A$.

Given a set $A$, a \emph{valuation in $A$} is any map $\rho : \Var\to A$. We write $\Val A$ for the set of all valuations in $A$. Given $\rho\in\Val A$ and $a\in A$, define:
\[
	(\rho\repl{x}{a})(y) = \begin{cases}
	a,&\textrm{if }x=y,\\
	\rho(y),&\textrm{otherwise}.
	\end{cases}
\]

\begin{defi}\label{def:syntmod}
A \emph{syntactic \lam-model} is a tuple $\cS = (A,\cdot,\Int{-}{-})$ such that $(A,\cdot)$ is an applicative structure and the \emph{interpretation function}
\[
	\Int{-}{-} : \Lam(A)\times \Val{A}\to A
\]
satisfies
\bsub
\item\label{def:syntmod1}
	$\Int{x}{\rho} = \rho(x)$, for all $x\in\Var$;
\item\label{def:syntmod2}
	$\Int{\cons a}{\rho} = a$, for all $a\in A$;
\item\label{def:syntmod3}
	$\Int{PQ}{\rho} = \Int{P}{\rho}\cdot \Int{Q}{\rho}$;
\item\label{def:syntmod4}
	$\Int{\lam x.P}{\rho}\cdot a = \Int{P}{\rho\repl{x}{a}}$, for all $a\in A$;
\item\label{def:syntmod5}
	$\forall x\in\FV{M}\,.\, \rho(x) = \rho'(x) \imp \Int{M}{\rho} = \Int{M}{\rho'}$;
\item\label{def:syntmod6}
	$\forall a\in A\,.\, \Int{M}{\rho\repl{x}{a}} = \Int{N}{\rho\repl{x}{a}} \imp \Int{\lam x.M}{\rho} = \Int{\lam x.N}{\rho}$.
\esub

\noindent
If $M\in\Lamo$, then $\Int{M}{\rho}$ is independent from the valuation $\rho$ and we simply write $\Int{M}{}$.

We write $\cS \models M = N$ if and only if $\forall \rho\in\Val A\,.\, \Int{M}{\rho} = \Int{N}{\rho}$ holds.
It is easy to check that $\blam\vdash M = N$ entails $\cS\models M = N$.
\end{defi}

The \emph{\lam-theory induced by $\cS$} is defined as follows:
\[
	\Th{\cS} = \set{ M = N \st \cS\models M = N}.
\]
The precise correspondence between \lam-models and syntactic \lam-models is described in~\cite{Bare}, Theorem~5.3.6. For our purposes, it is enough to know that if $\cS$ is a syntactic \lam-model then $\cC_\cS = (A,\cdot,\Int{\comb{K}}{},\Int{\comb{S}}{})$ is a \lam-model. We say that $\cS$ is \emph{extensional} whenever $\cC_\cS$ is extensional as a combinatory algebra.
This holds iff $\Th{\cS}$ is extensional iff $\cS \models \comb{I} = \comb{1}$.

\section{Addressing Machines}\label{sec:machines}

In this section we introduce the notion of an \emph{Addressing Machine}.
We first provide some intuitions, then we proceed with the formal description of such machines.
The general structure of an addressing machine is composed by two substructures:
\begin{itemize}
\item the \emph{internal components}, organized as follows:
	\begin{itemize}
	\item a finite number of \emph{internal registers};
	\item an \emph{internal program}.
\end{itemize}
\item the \emph{input-tape}.
\end{itemize}
As the name suggests, the addressing mechanism is central in this formalism.
Each addressing machine is associated with an address, receives a list of addresses in its input-tape and is able to transfer the computation to another machine by calling its address, possibly extending its input-tape.

\subsection{Tapes, Registers and Programs}
We consider fixed a countable set $\Addrs$ of \emph{addresses}, together with a constant $\Null\notin\Addrs$ that we call ``null'' and that corresponds to an uninitialized register.
\begin{defi} We let $\Addrs_\Null = \Addrs\cup\set{\Null}$.
\bsub
\item
	An \emph{$\Addrs$-valued tape} $T$ is a finite (possibly empty) ordered list of addresses $T = [a_1,\dots,a_n]$ with $a_i\in\Addrs$ for all $i \le n$.
	We write $\Tapes$ for the set of all $\Addrs$-valued tapes.
\item
	 Let $a\in\Addrs$ and $T,T'\in\Tapes$. We denote by $\Cons a {T}$ the tape having $a$ as first element and $T$ as tail. We write $\appT{T}{T'}$ for the concatenation of $T$ and $T'$, which is an $\Addrs$-valued tape itself.

\item
	Given an index $i\in\nat$, an $\Addrs_\Null$-valued \emph{register} $R_i$ is a memory-cell capable of storing either $\Null$ or an address $a\in\Addrs$.
 \item Given $\Addrs_\Null$-valued registers $R_0,\dots,R_{n}$ for $n\ge 0$, an address $a\in\Addrs$ and an index $i\in\nat$, we write $\vec R\repl{R_i}{a}$ for the registers $\vec R$ where the value of $R_i$ has been updated:
 \[
 R_0,\dots,R_{i-1},a,R_{i+1},\dots,R_{n}
 \]
Notice that, whenever $i > n$, we assume that $\vec R\repl{R_i}{a} = \vec R$.
\esub
\end{defi}

\noindent
Addressing machines can be seen as having a RISC architecture, since their internal program is composed by only three instructions. We describe the effects of these basic operations on a machine having $r$ internal registers $R_0,\dots,R_{r-1}$.
Therefore, when we say ``if an internal register $R_i$ exists'' we mean that the condition $0\le i< r$ is satisfied.
In the following, $i,j,k\in\nat$ correspond to indices of internal registers:
	\begin{itemize}
	\item $\Load i$: corresponds to the action of reading the first element $a$ from the input-tape $T$, and writing $a$ on the internal register $R_i$. If the input-tape is empty then the machine remains stuck waiting for an input (however, this is not considered as an error state).\\[3pt]
The \emph{precondition} to execute the operation is that the input-tape is non-empty, namely $T = \Cons a T'$; the \emph{postconditions} are that $R_i$, if it exists, contains the address $a$ and the input-tape of the machine becomes $T'$.
	If $R_i$ does not exist, i.e.\ when $i\ge r$, the content of $\vec R$ remains unchanged (i.e., the input element $a$ is read and subsequently thrown away).
	\item $\Apply i j k$: corresponds to the action of reading the contents of $R_i$ and $R_j$, calling an external \emph{application map} on the corresponding addresses $a_1,a_2$, and writing the result in the internal register $R_k$, if it exists.\\[3pt]
The \emph{precondition} is that $R_i,R_j$ exist and are initialized, i.e.\ $R_i,R_j\neq\Null$.
The \emph{postcondition} is that $R_k$, if it exists, contains the address of the machine of address $a_1$ whose input-tape has been extended with $a_2$.
Otherwise the content of $\vec R$ remains unchanged.
	\item
	$\Call i$: transfers the computation to the machine whose address is stored in $R_i$, extending its input-tape with the addresses that are left in $T$.\\[3pt]
	The \emph{precondition} is that $R_i$ exists and is initialized.
	The \emph{postcondition} is that the machine having the address stored in $R_i$ is executed on the extended input-tape.
	\end{itemize}

\noindent
We define what is a syntactically valid program of this language, and  introduce a decision procedure for verifying that the preconditions of each instruction are satisfied when it is executed.
As we will see in Lemma~\ref{lem:correction}, these properties are decidable and statically verifiable.
As a consequence, addressing machines will never give rise to an error at run-time.

\begin{defi}\label{def:progs}\
\bsub
\item\label{def:progs1}
	A \emph{program} $P$ is a finite list of instructions generated by the following grammar (where $\varepsilon$ represents the empty string, and $i,j,k\in\nat$):
	\[
	\begin{array}{lcl}
	\ins{P}&\eqbnf&\Load i;\, \ins{P}\mid \ins{A}\\
	\ins{A}&\eqbnf&\Apply ijk;\, \ins{A}\mid \ins{C}\\
	\ins{C}&\eqbnf&\Call i \mid \varepsilon
	\end{array}
	\]
	In other words a program starts with a list of $\ins{Load}$'s, continues with a list of $\ins{App}$'s and possibly ends with a $\ins{Call}$. Each of these lists may be empty, in particular the empty-program $\varepsilon$ can be generated.
\item\label{def:progs2}
	Given a program $P$, an $r\in\nat$, and a set $\cI\subseteq \set{0,\dots,r-1}$ of indices (representing initialized registers), define $\cI\models^{r} P$ as the least relation closed under the rules:
\[
	\begin{array}{ccccc}
		\infer{\cI\models^{r}\varepsilon}{}
		&&
		\infer{\cI\models^{r} \Apply ijk;\, \ins{A}}{\cI\cup\set{k}\models^{r}  \ins{A} & i,j\in \cI & k<r}
		&&
		\infer{\cI\models^{r} \Load i;\, \ins{P}}{\cI\cup\set{i}\models^{r}  \ins{P} & i< r}
		\\[3pt]
		\infer{\cI\models^{r}\Call i}{i\in \cI}&&
		\infer{\cI\models^{r} \Apply ijk;\, \ins{A}}{\cI\models^{r}  \ins{A} & i,j\in \cI & k\ge r}
		&&
		\infer{\cI\models^{r} \Load i;\, \ins{P}}{\cI\models^{r}  \ins{P} & i \ge r}
	\end{array}
\]
\item Let $r\in\nat$ and $\vec R = R_0,\dots,R_{r-1}$ be $\Addrs_\Null$-valued registers.
	We say that a program $P$ is \emph{valid with respect to $\vec R$} whenever $\cR\models^{r} P$ holds for
	\begin{equation}\label{eq:R}
		\cR = \set {i\st R_i \neq\Null \und 0\le i < r}
	\end{equation}
\esub
\end{defi}

\noindent
Notice that the notion of a valid program is independent from the tape of a machine.

\begin{exas} Consider addresses $a_1, a_2\in\Addrs $, as well as $\Addrs_\Null$-valued registers $R_0 = \Null$, $R_1 = a_1,R_2=a_2,R_3 = \Null$ (so $r = 4$).
In this example, the set $\cR$ of initialized registers as defined in~\eqref{eq:R} is $\cR = \set{1,2}$.
\[
	\begin{array}{lcc}
	P_n&\textrm{Program}&\cR\models^4 P_n\\
	\toprule
	P_0=&\Load 0;\Apply012;\Call 2&\checkmark\\
	P_1=&\Apply 120;\,\Apply 023;\,\Call 3&\checkmark\\
	P_2=&\Load 5;\, \Load 0;\,\Call 0&\checkmark\\
	P_3=&\Load 5;\, \Apply 12{5};\,\Call 2&\checkmark\\
	P_4=&\Apply 012;\,\Call 2&\xmark\\
	P_5=&\Load 0;\,\Call 3&\xmark\\
	P_6=&\Apply 123;\,\Call 5&\xmark\\
	\end{array}
\]
Above we use ``5'' as an index of an unexisting register.
Notice that a program trying to update an unexisting register remains valid (see $P_2,P_3$), the new value is simply discharged.
On the contrary, an attempt at reading the content of an uninitialized ($P_4,P_5$) or unexisting ($P_6$) register invalidates the whole program.
\end{exas}

\begin{nota}\label{nota:aboutprogs}
We use ``$-$'' to indicate an arbitrary index of an unexisting   register. E.g., the program $P_6$ will be written $\Apply 123;\,\Call -$.
We also write $\Load (i_1,\dots,i_k)$ as an abbreviation for $\Load i_1;\,\cdots;\,\Load i_k;$ . By employing all these notations, $P_2$ can be written as  $P_2= \Load (-,0);\Call 0$. 
\end{nota}

\begin{lem}\label{lem:correction}
For all $\Addrs_\Null$-valued registers $\vec R$ and program $P$ it is decidable whether $P$ is valid with respect to $\vec R$.
\end{lem}

\begin{proof}
Decidability follows from the fact that the grammar in Definition~\ref{def:progs}\ref{def:progs1} is right-linear, the list of registers $\vec R$ is finite, the rules in Definition~\ref{def:progs}\ref{def:progs2} are syntax-directed and their side conditions are decidable.
\end{proof}

\subsection{Addressing machines and their operational semantics}

Everything is in place to introduce the definition of an \am.
Thanks to Lemma~\ref{lem:correction} it is reasonable to require that an \am{} has a valid internal program.

\begin{defi}\label{def:AM}
\bsub
\item
	An \emph{addressing machine $\mM$ (with $r$ registers) over $\Addrs$} is given by a tuple:
\[
	\mM = \tuple{\vec R,P,T}
\] where:
\begin{itemize}
\item
	$\vec R = R_0,\dots,R_{r-1}$ are $\Addrs_\Null$-valued registers;
\item
	$P$ is a program valid w.r.t.\ $\vec R$;
\item
	$T$ is an $\Addrs$-valued \emph{(input) tape}.
\end{itemize}
\item
	We write $\mM.r$ for the number of registers of $\mM$, $\mM.\vec R$ for the list of its registers, $\mM.R_i$ for its $i$-th register, $\mM.P$ for the associated program and finally $\mM.T$ for its input tape.
\item
	We say that an addressing machine $\mM$ as above is \emph{stuck}, in symbols $\stuck{\mM}$, whenever its program has shape $\mM.P = \Load i;\ins{P}$ but its input-tape is empty $\mM.T = []$. Otherwise, $\mM$ is \emph{not stuck}, in symbols: $\lnot\stuck{\mM}$.
\item
	The set of all addressing machines over $\Addrs$ will be denoted by $\cM$.
\esub
\end{defi}

\noindent
The machines below will be used as running examples in the next sections.
Intuitively, the \am s $\mK,\mS,\mach{I},\mach{D},\mach{O}$ mimic the behavior of the \lam-terms $\comb{K}$, $\comb{S}$, $\comb{I}$, $\comb{\Delta}$ and $\comb{\Omega}$, respectively. For writing their programs, we adopt the conventions introduced in Notation~\ref{nota:aboutprogs}.

\begin{exas}\label{ex:ilprimoesempiononsiscordamai}
The following are addressing machines.
\bsub
\item\label{ex:ilprimoesempiononsiscordamai1}
	For every $n\in\nat$, define an addressing machine with $n+1$ registers as:
\[
	\mach{x}_n = \tuple{R_0,\dots,R_n,\varepsilon,[]},\textrm{ where }\vec R := \vec \Null.
\]
We call $\mach{x}_0,\mach{x}_1,\mach{x}_2,\dots$ \emph{indeterminate machines} because they share some analogies with variables (they can be used as place holders).
\item
The addressing machine $\mK$ with 1 register $R_0$ is defined by:
\[
	\mK = \tuple{\Null,\RaS (0,-); \Call 0,[]}
\]
\item The addressing machine $\mS$ with 3 registers is defined by:
\[
	\begin{array}{lcl}
	\mS &=& \tuple{\Null,\Null,\Null,P,[]}\textrm{, where:}\\
	\mS.P &=& \RaS (0, 1,2);\,\Apply 020;\\
	&&\Apply 121;\,\Apply 012;\,\Call 2\\
	\end{array}
\]
\item Assume that $k\in\Addrs$ represents the address associated with the \am{} $\mK$.
Define the \am{} $\mach{I}$ as $\mach{I} = \tuple{\Null^3,\mS.P,[k,k]}$.
\item The addressing machine $\mach{D}$ with 1 register is given by:
\[
	\mach{D} = \tuple{\Null,\RaS 0;\,\Apply 000;\,\Call 0,[]}
\]
\item Assume that $d\in\Addrs$ represents the address of the \am{} $\mach{D}$.
Define the \am{} $\mach{O}$ by setting $\mach{O} = \tuple{\Null,\mach{D}.P,[d]}$.
\esub
\end{exas}

\noindent
We now enter into the details of the addressing mechanism which constitutes the core of this formalism.
In an implementation of \am s, it would be reasonable to pick up a fresh address from $\Addrs$ whenever a new machine is constructed and save the correspondence in some address table. See Section~\ref{sec:conclusions} for more implementation details.
To construct a \lam-model, we need a uniform way of associating machines with their addresses.

\begin{defi} Fix a bijective map $\Lookup : \cM \to  \Addrs$ from the set of all \am s over $\Addrs$ to the set $\Addrs$ of addresses.
We call the map $\Lookup(\cdot)$ an \emph{Address Table Map (ATM)}.
\bsub
\item Given $M\in\cM$, we say that $\Lookup \mM$ is the \emph{address of} $\mM$.
\item
	Given an address $a\in\Addrs$, we write $\Lookinv{a}$ for the unique machine having address $a$. In other words, we have $\Lookinv{a} = \mM\iff \Lookup\mM = a$.
\item
	Given $\mM\in\cM$ and $T'\in\Tapes$, we write $\appT{\mM}{T'}$ for the machine
	\[
		\tuple{\mM.\vec R,\mM.P,\appT{\mM.T}{T'}}
	\]
\item
	Define the \emph{application map} $(\App{}{}) : \Addrs\times\Addrs\to \Addrs$ as follows
	\[
		\App{a}{b} = \Lookup (\append{\Lookinv{a}}{b})
	\]
	That is, the \emph{application} of $a$ to $b$ is the unique address $c$ of the \am{} obtained by adding $b$ at the end of the input tape of the \am{} $\Lookinv{a}$.
\esub
\end{defi}

\noindent
Since both $\cM$ and $\Addrs$ are countable sets, there exist $2^{\aleph_0}$ possible choices for an ATM\@.
\begin{rem}\label{rem:forever} Depending on the chosen ATM $\Lookup(-)$, there might exist \am s calling each other, as in $\mM = \tuple{\Lookup\mN,\Call 0,[]}$ and $\mN = \tuple{\Lookup\mM,\Call 0,[]}$, or even countably many machines $(\mM_n)_{n\in\nat}$ satisfying $\mM_n = \tuple{\Lookup\mM_{n+1},\varepsilon,[]}$.
Therefore, in general, the process of recursively dereferencing the addresses stored in the registers (or tape) of a machine might not terminate.
This kind of behaviour is not pathological, rather intrinsic to the notions of addresses and dereference operators.
\end{rem}
In practice, one may desire to work with an ATM performing the association between \am s and their addresses in a computable way.
However, we do not require our ATMs to satisfy any effectiveness conditions since it would be peculiar to propose a model of computation depending on a pre-existing notion of ``computable''. The results presented in this paper are independent from the ATM  under consideration.

\begin{defi}[Small step operational semantics]\label{def:smallstep}
Define a reduction strategy on \am s representing one head-step of computation
\[
	\redh\ \subseteq\cM\to\cM
\]
as the least relation closed under the following rules:
\[
	\begin{array}{lcl}
	\tuple{\vec R,\RaS i;P,\Cons a{T}} &\redh& \tuple{\vec R[R_i := a],P,T},\\
	\tuple{\vec R,\Apply i j k; P,T}&\redh&\tuple{\vec R[R_k := \App{R_i}{R_j}],P,T},\\
	\tuple{\vec R,\Call i,T}&\redh&\appT{\Lookinv {R_i}}{T}.\\
	\end{array}
\]
As usual, we write $\reddh$ for the transitive-reflexive closure of $\redh$.
We say that an \am{} $\mM$ \emph{is in a final state} if there is no $\mN$ such that $\mM\redh \mN$.
We write $\mM\reddh \stuck{\mN}$ whenever $\mM\reddh \mN$ and $\stuck{\mN}$ hold.
When $\mN$ is not important, we simply write $\mM\reddh\stuck{}$. Similarly, $\mM\not\reddh\stuck{}$ means that $\mM$ never reduces to a stuck \am.
\end{defi}

\begin{rem}\label{rem:aboutstuck}\
\bsub
\item Definition~\ref{def:smallstep} is well defined since the validity of a program is preserved by ${\sf h}$-reduction: if $\mM\redh\mN$ and $\mM.P$ is valid w.r.t.\ $\mM.\vec R$ then $\mN.P$ is valid w.r.t.\ $\mN.\vec R$. This follows immediately from Definition~\ref{def:progs}\ref{def:progs2}.
In particular when executing $\RaS i$, or $\Call i$, $R_i$ must be initialized and when executing $\Apply i j k$ we must have $R_i,R_j\neq\Null$.
\item\label{rem:aboutstuck2}
Addressing machines in a final state are either of the form $\tuple{\vec R,\varepsilon,T}$ or $\tuple{\vec R,\Load i;P,[]}$, and in the latter case they are stuck.
\esub
\end{rem}

\begin{lem}\label{lem:about_red}
The reduction strategy $\redh$ enjoys the following properties:
\bsub
\item\label{lem:about_red1}
	 Determinism: $\mM\redh \mN_1 \und \mM\redh \mN_2\ \Rightarrow\ \mN_1 = \mN_2$.
\item\label{lem:about_red2}
	Closure under application: $\forall a\in\Addrs\,.\,\mM\redh \mN\ \Rightarrow\ \append{\mM}{a}\redh \append{\mN}{a}$.
\esub
\end{lem}
\begin{proof} $(i)$ Since the applicable rule from Definition~\ref{def:smallstep}, if any, is uniquely determined by the first instruction on $\mM.P$ and its input-tape $\mM.T$.

$(ii)$ Easy. By cases on the rule applied for deriving $\mM\redh\mN$.
\end{proof}

\begin{exas}\label{ex:somemachines}
For brevity, we sometimes display only the first instruction of the internal program. Take $a,b,c\in\Addrs$.
\bsub
\item We show that $\mach K$ behaves as the first projection:
\[
	\begin{array}{lll}
	\append{\mK}{a,b}&=&\tuple{\Null,\RaS (0, -); \Call 0,[a,b]}\\
	&\redh&\tuple{a, \RaS -; \Call 0,[b]}
	\redh
	\tuple{a, \Call 0,[]}\redh\Lookinv a.\\
	\end{array}
\]
\item We verify that $\mach S$ behaves as the combinator $\comb{S}$ from combinatory logic:
\[
	\begin{array}{lll}
	\append{\mS}{a,b,c}&=&\tuple{\Null^3,\RaS (0,1,2); \cdots,[a,b,c]}\\
	&\reddh&\tuple{a,b,c,\Apply 020; \cdots,[]}\\
	&\redh&\tuple{\App{a}{c},b,c,\Apply 121; \cdots,[]}\\
	&\redh&\tuple{\App{a}{c},\App{b}{c},c,\Apply 012; \cdots,[]}\\
	&\redh&\tuple{\App{a}{c},\App{b}{c},\App{(\App{a}{c})}{(\App{b}{c})},\Call 2; \cdots,[]}\\
	&\redh&\Lookinv {\App{(\App{a}{c})}{(\App{b}{c})}}\\
	\end{array}
\]
\item As expected, $\mach{I}=\append{\mS}{\Lookup \mK,\Lookup \mK}$ behaves as the identity:
\[
	\begin{array}{lll}
	\append{\mach{I}}{a} &=&
	\tuple{\Null^3,\RaS (0,1,2);\cdots,[\Lookup{\mK},\Lookup{\mK},a]}\\
	&\reddh&\tuple{\Lookup{\mK},\Lookup{\mK},a,\Apply 020;\cdots,[]}\\
	&\redh&\tuple{\App{\Lookup{\mK}}{a},\Lookup{\mK},a,\Apply 121;\cdots,[]}\\
	&\redh&\tuple{\App{\Lookup{\mK}}{a},\App{\Lookup{\mK}}{a},a,\Apply 012;\cdots,[]}\\
	&\redh&\tuple{\App{\Lookup{\mK}}{a},\App{\Lookup{\mK}}{a},\App{\App{\Lookup{\mK}}{a}}{(\App{\Lookup{\mK}}{a})},\Call 2;[]}\\
	&\redh&\append{\mK}{a,\App{\Lookup{\mK}}{a}}\\
	&=&\tuple{\Null,\RaS (0,-);\cdots,[a,\App{\Lookup{\mK}}{a}]}\\
	&\reddh&\tuple{a,\App{\Lookup{\mK}}{a},\Call 0,[]}
	\redh\Lookinv{a}\\
	\end{array}
\]
\item Finally, we check that $\mach{O}$ gives rise to an infinite reduction sequence:
\[
	\begin{array}{lll}
	\mach{O} &=& \tuple{\Null,\RaS 0;\,\Apply 000;\,\Call 0,[\Lookup\mach{D}]}\\
	&\redh&\tuple{\Lookup\mach{D},\Apply 000;\,\Call 0,[]}\\
	&\redh&\tuple{\Lookup(\append{\mach{D}}{\Lookup\mach{D}}),\Call 0,[]}\redh \append{\mach{D}}{\Lookup\mach{D}}
	= \mach{O}\reddh\cdots\\
	\end{array}
\]
\esub
\end{exas}

\noindent
Similarly, we can define a big-step operational semantics relating an \am{} $\mM$ with its final result (if any).

\begin{defi}[Big-step semantics]
Define $\mM \goesto \mach{V}$, where $\mM,\mV\in\cM$ and $\mV$ is in a final state, as the least relation closed under the following rules:
\begin{gather*}
	\infer[(\textrm{Stuck})]{\mM \goesto \mM}{
		\mM.P = \RaS i;P'
		&
		\mM.T = []
		}
	\qquad\qquad
	\infer[(\textrm{End})]{\mM \goesto \mM}{
		\mM.P = \varepsilon
		}\\
		\infer[(\textrm{Load})]{\mM \goesto \mV}{
		\mM.P = \RaS i;P'
		&
		\mM.T = \Cons a{T'}
		&
		\tuple{\mM.\vec R\repl{R_i}{a},P',T'}\goesto \mV
		}
		\\
		\infer[(\textrm{App})]{
		\mM\goesto\mV}{\mM.P = \Apply i j k; P'
		&
		a = \App{\mM.R_i}{\mM.R_j}
		&
		\tuple{\mM.\vec R\repl{R_k}{a},\mM.P',\mM.T}\goesto\mV
	}
	\\
		\infer[(\textrm{Call})]{\mM\goesto \mV}{
			\mM.P = \Call i
			&
			\mM' = \Lookup^{-1}(\mM.R_i)
			&
			\append{\mM'}{\mM.T}\goesto \mV
	}
\end{gather*}

\begin{exa} Recall that $\mK.P = \Load (0,-);\,\Call 0$. Notice that we cannot prove $\append{\mK}{a,b}\goesto \Lookinv a$ for an arbitrary $a\in\Addrs$, as we need to ensure that the resulting machine is in a final state.
For this reason, we will use indeterminate machines $\mach{x}_1,\mach{x}_2$ from Example~\ref{ex:ilprimoesempiononsiscordamai}\ref{ex:ilprimoesempiononsiscordamai1}.
\[
	\infer{\append{\mK}{\Lookup \mach{x}_1,\Lookup \mach{x}_2}\goesto \mach{x}_1}{
	\mK.P = \Load 0;\,P';
	&
	\infer{\tuple{\Lookup\mach{x}_1,P',[\Lookup\mach{x}_2]}\goesto \mach{x}_1}{
		P'= \Load -;P'' &
		\infer{\tuple{\Lookup\mach{x}_1,P'',[]}\goesto \mach{x}_1}{
				P''=\Call 0
				&
				R_0 = \Lookup\mach{x}_1
				&
				\infer{\mach{x}_1\goesto \mach{x}_1}{\textrm{(End)}}
		}
	}
	&
	}
\]
\end{exa}
\end{defi}

We now show that the two operational semantics are equivalent on terminating computations.

\begin{prop}\label{prop:equivsem}
For $\mM,\mN\in\cM$, the following are equivalent:
\begin{enumerate}
\item $\mM\reddh \mN\not\redh$;
\item $\mM\goesto\mN$.
\end{enumerate}
\end{prop}

\begin{proof}
(1 $\Rightarrow$ 2) By induction on the length $n$ of the reduction $\mM=\mM_1\redh\mM_2\redh\cdots\redh \mM_n=\mN\not\redh$.

Case $n = 0$. By assumption $\mN$ is in a final state. By Remark~\ref{rem:aboutstuck}\ref{rem:aboutstuck2}, it is either of the form $\mN = \tuple{\vec R,\varepsilon,T}$ or it is stuck $\mN = \tuple{\vec R,\Load i;P,[]}$. In the former case we apply (\textrm{End}), in the latter (\textrm{Stuck}).

Case $n > 1$. Since $\mM_1\redh \mM_2$, we have $\mM_1.P\neq\varepsilon$.
As the length of $\mM_2\reddh \mN$ is $n-1$, by induction hypothesis we have a derivation of $\mM_2\goesto \mN$.
Depending on the first instruction in $\mM_1.P$, we use this derivation to apply the homonymous rule (Load), (App) or (Call) and derive $\mM\goesto \mN$.

(2 $\Rightarrow$ 1) By induction on a derivation of $\mM\goesto \mN$.

Cases (Stuck) or (End). Then, $\mM\reddh\mM=\mN$ by reflexivity of $\reddh$.

Case (Load), i.e.\ $\mM.P=\Load i;P'$. In this case, we have that $\mM\redh \tuple{\mM.\vec R\repl{R_i}{a},P',\mM.T}\reddh \mN$, by induction hypothesis.

Case (App), i.e.\ $\mM.P=\Apply ijk;P'$. Let us call $a = \App{\mM.R_j}{\mM.R_k}$. Then we have $\mM\redh \tuple{\mM.\vec R\repl{R_k}{a},P',\mM.T}\reddh \mN$, by induction hypothesis.

Case (Call), i.e.\ $\mM.P = \Call i$. In this case $\mM\redh \append{\mM'}{\mM.T}$ for $\mM' = \Lookinv{\mM.R_i}$. By induction hypothesis $\append{\mM'}{\mM.T}\reddh \mN$, whence $\mM\reddh \mN$.
\end{proof}

\section{Combinatory Algebras via Evaluation Equivalence}\label{sec:combalg}

In this section we show how to construct a combinatory algebra based on the \am s formalism. 
Recall that the \am s $\mK$ and $\mS$ have been defined in Example~\ref{ex:ilprimoesempiononsiscordamai}. Consider the algebraic structure 
\[
	\cA = (\Addrs,\App{\,}{\,},\Lookup\mach{K},\Lookup\mS)
\]
Since the application $(\App{}{})$ is total, $\cA$ is an applicative structure.
However, it is \emph{not} a combinatory algebra.
For instance, the $\lama$-term $\comb{K}\cons a\cons b$ is interpreted as the address of the machine $\append{\mach{K}}{a,b}$, which is \emph{a priori} different from the address ``$a$'' because no computation is involved.
Therefore, we need to quotient the algebra $\cA$ by an equivalence relation equating at least all addresses corresponding to the same machine at different stages of the execution.

In the following, we denote by $\equiv_{\rel R}$ an arbitrary binary relation on $\cM$. The symbol ${\rel R}$ has no formal meaning, it is simply evocative of a relation.
In the next definition, we are going to associate with every $\equiv_{\rel R}$ two relations, respectively denoted $\simeq_{\rel R}\,\subseteq \Addrs^2$ and $=_{\rel R}\,\subseteq \cM^2$.

\begin{defi}\label{def:inducingequivalences}
Every binary relation $\equiv_{\rel R}\,\subseteq\cM^2$ on \am s induces a relation $\simeq_{\rel R}\,\subseteq \Addrs^2$ defined by 
\[
	a\simeq_{\rel R} b\iff \Lookinv{a}\equiv_{\rel R} \Lookinv{b}
\]
which is then extended to:
\bsub
\item $\Addrs_\Null$-valued registers:
\[
	R\simeq_{\rel R} R' \iff (R = \Null = R')\lor (R = a \simeq_{\rel R} b =R');
\]
\item Tuples:
\[
	a_1,\dots,a_n \simeq_{\rel R} b_1,\dots,b_m \iff (n = m)\land (\forall i\in\set{1,\dots,n}\,.\, a_i\simeq_{\rel R} b_i);
\]
(This also applies to tuples of $\Addrs_\Null$-valued registers $\vec R\simeq_{\rel R} \vec R'$.)
\item  $\Addrs$-valued tapes:
	\[
		[a_1,\dots,a_n]\simeq_{\rel R} [b_1,\dots,b_m]\iff \vec a \simeq_{\rel R} \vec b \textrm{ (seen as tuples).}
	\]
\esub
\noindent
In its turn, $\simeq_{\rel R}$ induces a relation $=_{\rel R}\ \subseteq\cM^2$ defined by setting (for all machines $\mM,\mN\in\cM$):
\[
	\mM =_\rel{R} \mN \iff (\mM.\vec R \simeq_{\rel R}\mN.\vec R)\land (\mM.P = \mN.P)\land(\mM.T \simeq_{\rel R}\mN.T)
\]
\end{defi}

In particular, $\mM =_\rel{R} \mN$ entails that $\mM$ and $\mN$ share the same internal program, the number of internal registers, and the length of their input tape.

\begin{lem}\label{lem:equivalence}
If the relation $\equiv_\rel{R}$ is an equivalence then so are $\simeq_\rel{R}$ and $=_\rel{R}$.
\end{lem}
\begin{proof} Assume that $\equiv_\rel{R}$ is an equivalence. Then, the fact that $\simeq_\rel{R}$ is an equivalence follows from its definition since $\Lookinv{\cdot}$ is a bijection. Concerning the relation $=_\rel{R}$, reflexivity, symmetry and transitivity follow immediately from the same properties of $\simeq_{\rel{R}}$ and $=$.
\end{proof}

\begin{defi}\label{def:equiv:Addrs}
Define $\equiva\ \subseteq\cM^2$ as the least equivalence closed under:
\[
	\infer[\redrule]{\mM\equiv_\Addrs	 \mN}{\mM\reddh \mZ =_\Addrs \mN}
\]
\end{defi}
We say that $\mM,\mN$ are \emph{evaluation equivalent} whenever $\mM\equiva\mN$.

\begin{rem}\
\begin{enumerate}[(i)]
\item Reflexivity can be treated as a special case of the rule $\redrule$ since $\mM\reddh\mM=_\Addrs \mM$.
\item It follows from the definition that $=_\Addrs\,\subseteq\ \equiva$ and that $\mM\reddh \mN$ entails $\mM\equiva \mN$.
\end{enumerate}
\end{rem}

\begin{exas}\label{ex:calculs}
From the calculations in Examples~\ref{ex:somemachines}, it follows that
\[
	\begin{array}{lcl}
	\append{\mK}{\Lookup\mach{x}_1,\Lookup\mach{x}_2}&\equiva& \mach{x}_1, \\
	\append{\mS}{\Lookup\mach{x}_1,\Lookup\mach{x}_2,\Lookup\mach{x}_3}&\equiva& \append{(\append{\mach{x}_1}{\Lookup \mach{x}_3})}{\Lookup(\append{\mach{x}_2}{\Lookup \mach{x}_3})}.\\
	\end{array}
\]
\end{exas}

\begin{lem}
The relation $\sima$ is a congruence on $\cA= (\Addrs,\App{\,}{\,},\Lookup\mach{K},\Lookup\mS)$.
\end{lem}
\begin{proof}
By definition $\equiva$ is an equivalence, whence so is $\sima$ by Lemma~\ref{lem:equivalence}.
Let us check that $\sima$ is compatible w.r.t.\ $(\App{}{})$.
Consider $a \sima a'$ and $b\sima b'$.
Call $\mM = \Lookinv{a}$ and $\mN = \Lookinv{a'}$ and proceed by induction on a derivation of $\mM\equiva\mN$, splitting into cases depending on the last applied rule.

\redrule{} By definition, there exists $\mZ\in\cM$ such that $\mM\reddh \mach Z =_\Addrs \mN$. By Lemma~\ref{lem:about_red}\ref{lem:about_red2}, $\append{\mM}{b} \reddh \append{\mZ}{b} =_\Addrs \append{\mN}{b'}$ whence $\App{a}{b}\sima\App{a'}{b'}$.

(Transitivity) and (Symmetry) follow from the induction hypothesis.
\end{proof}

In order to prove that the congruence $\sima$ is non-trivial, we are going to characterize the equivalence $\mM\equiva\mN$ it in terms of confluent reductions.
For this purpose, we extend $\redh$ in such a way that reductions are also possible within registers and elements of the input-tape of an \am.

\begin{defi} Define the reduction relation $\red[c]\,\subseteq\cM^2$ as the least relation containing $\redh$ and closed under the following rules:
\begin{gather*}
\infer[{(\red[i]^R)}]{\tuple{R_0,\dots,R_{r-1},P,T} \red[c] \tuple{\vec R\repl{R_i}{\Lookup\mM},P,T}}{R_i = a\in\Addrs&0\le i<r& \Lookinv{a}\red \mM}\\[3pt]
\infer[{(\red[i]^T)}]{\tuple{\vec R,P,[a_0,\dots,a_n]} \red[c] \tuple{\vec R,P,[a_0,\dots,a_{i-1},\Lookup\mM,a_{i+1},\dots,a_n]}}{0\le i\le n& \Lookinv{a_i}\red \mM}
\end{gather*}
We write $\mM\red[i]\mN$ if $\mN$ is obtained from $\mM$ by directly applying one of the above rules --- this is called an \emph{inner} step of computation.
The transitive and reflexive closure of $\red$ and $\red[i]$ are denoted by $\redd$ and $\redd[i]$, respectively.
\end{defi}

\begin{lem}[Postponement of inner steps]\label{lem:standardization}~\\ For $\mM,\mN,\mN'\in\cM$, if $\mM\red[i]\mN\red[h]\mN'$ then there exists $\mM'\in\cM$ such that $\mM\red[h]\mM'\redd[i]\mN'$. In diagrammatic form:
\[
\xymatrix{
\mM\ar@{->}[r]^{\mach{i}}\ar@{-->}[d]^{\mach{h}}&\mN\ar@{->}[d]^{\mach{h}}\\
\mM'\ar@{-->>}[r]^{\mach{i}}&\mN'
}
\]
\end{lem}

\begin{proof} By cases analysis over $\mM\red[i]\mN$.
The only interesting case is when the contracted redex is duplicated in $\mN\red[h]\mN'$, namely:

Case $\mM =\tuple{\vec R\repl{R_i}{a},P,T}$, $\mN = \tuple{\vec R\repl{R_i}{b},P,T}$ with  $\mM.P =\mN.P = \Apply ijk;P'$ and $\Lookinv a\red[c]\Lookinv b$.
Assume $i\neq k<\mM.r$ and $i = j$, the other cases being easier.
In this case $\mM' = \tuple{\vec R\repl{R_i}{a}\repl{R_k}{\App{a}{a}},P,T}$, therefore we need 3 inner steps to close the diagram:
\[
	\begin{array}{lcl}
	\mM'&\red[i]&\tuple{\vec R\repl{R_i}{b}\repl{R_k}{\App{a}{a}},P,T}\\
		&\red[i]&\tuple{\vec R\repl{R_i}{b}\repl{R_k}{\App{b}{a}},P,T}\\
			&\red[i]&\tuple{\vec R\repl{R_i}{b}\repl{R_k}{\App{b}{b}},P,T} = \mN'.
	\end{array}
\]
This concludes the proof.
\end{proof}

Morally, the term rewriting system $(\cM,\red[c])$ is orthogonal because $(i)$ the reduction rules defining $\red[c]$ are non-overlapping as $\red[h]$ is deterministic, $(\red[i]^R)$ reduces a register and $(\red[i]^T)$ reduces  one element of the tape; $(ii)$ the terms on the left-hand side of the arrow are linear, as no equality among subterms is required.
Now, it is well-known that orthogonal TRS are confluent, but one cannot apply~\cite[Thm.4.3.4]{terese} directly since we are not exactly dealing with first-order terms (because of the presence of the encoding).

\begin{prop}\label{prop:confluence}
The reduction $\red[c]$ is confluent.
\end{prop}

\begin{proof}[Proof sketch] The Parallel Moves Lemma, which is the key property for proving Theorem~4.3.4 in~\cite{terese} generalizes easily. The rest of the proof follows.
\end{proof}

\begin{lem}\label{lem:onestepisfine}
Let $\mM,\mN\in\cM$.
\begin{enumerate}[(i)]
\item\label{lem:onestepisfine1}
	$\mM\red\mN$ entails $\mM\equiva\mN$.
\item\label{lem:onestepisfine2}
	$\mM\redd\mN$ entails $\mM\equiva\mN$.
\end{enumerate}
\end{lem}

\begin{proof}
\ref{lem:onestepisfine1} By induction on a derivation of $\mM\red\mN$. 

Base case $\mM\redh \mN$. Since $\equiva$ is an equivalence then so is $=_\Addrs$, by Lemma~\ref{lem:equivalence}.
In particular $=_\Addrs$ is reflexive, whence $\mN =_\Addrs \mN$. By Definition~\ref{def:equiv:Addrs}, we obtain $\mM\equiva \mN$.

Case ${(\red[i]^R)}$. Then $\mM = \tuple{\vec R[R_i:= \Lookup \mM'],P,T}$ and $\mN = \tuple{\vec R[R_i:= \Lookup \mN'],P,T}$ for some existing register $R_i$ and $\mM',\mN'\in\cM$ such that $\mM' \red \mN'$. By induction hypothesis we get $\mM' \equiva \mN'$, equivalently $\Lookup\mM' \sima \Lookup\mN'$.
From this and reflexivity, it follows $\vec R[R_i:= \Lookup \mM'] \sima \vec R[R_i:= \Lookup \mN']$, $P \sima P$ and $T \sima T$.
Thus $\mM =_\Addrs \mN$, so we conclude because $=_\Addrs\,\subseteq\  \equiva$.

Case ${(\red[i]^T)}$. In this case, we have
\[
	\begin{array}{lcl}
	\mM &=& \tuple{\vec R,P,[a_0,\dots,a_{i-1},\Lookup \mM',a_{i+1}\dots,a_n]}\\
	 \mN &=& \tuple{\vec R,P,[a_0,\dots,a_{i-1},\Lookup \mN',a_{i+1}\dots,a_n]}
	 \end{array}
\]
with $\mM' \red \mN'$. By induction hypothesis we get $\mM' \equiva \mN'$, equivalently $\Lookup\mM' \sima \Lookup\mN'$.
This entails $\mM.T \sima \mN.T$, from which it follows $\mM =_\Addrs \mN$. Conclude as above.

\ref{lem:onestepisfine2}  By induction on the length $n$ of the reduction $\mM\redd\mN$. 

Case $n=0$. Then $\mM = \mN$, so we get $\mM\equiva\mN$ by reflexivity.

Case $n>0$. Then $\mM\red\mM'\redd\mN$. By~\ref{lem:onestepisfine1}, we get $\mM \equiva \mM'$.
Since the reduction $\mM'\redd\mN$ is strictly shorter, the induction hypothesis gives $\mM' \equiva \mN$.
Conclude by transitivity.
\end{proof}

\begin{thm}\label{thm:CR}
For $\mM,\mN\in\cM$, we have:
\[
	\mM\equiva \mN\iff \exists \mZ\in\cM\,.\, \mM\redd\mZ\invredd[\mach{c}] \mN
\]
\end{thm}

\begin{proof} $(\Rightarrow)$ By induction on a derivation of $\mM\equiva \mN$.

\redrule{} Assume that $\mM\reddh\mZ=_\Addrs \mN$. From $\mZ=_\Addrs \mN$ we get that $\mZ.r = \mN.r$, $\mZ.\vec R\sima \mN.\vec R$, $\mZ.P = \mN.P$ and $\mZ.T\sima\mN.T$. Note that $\mZ.R_i =\Null$ iff $\mN.R_i = \Null$.
Let us call $\cR$ the set of indices $i$ of, say, $\mZ$ such that $\mZ.R_i\neq\Null$.
By assumption, for every $i\in\cR$, we have $\mZ.R_i =a_i,\mN.R_i = a'_i$ for $a_i\sima a'_i$. Equivalently, $\Lookinv{a_i}\equiva\Lookinv{a'_i}$ holds and its derivation is smaller than $\mM\equiva \mN$. By induction hypothesis, they have a common reduct $\Lookinv{a_i}\redd \mach{X}_i\invredd[\mach{c}]\Lookinv{a'_i}$.
Similarly, calling $\mZ.T = [b_1,\dots,b_n]$ and $\mN.T = [b'_1,\dots,b'_m]$ we must have $m = n$ and $b_j\sima b'_j$ whence the induction hypothesis gives a common reduct $\Lookinv{b_j}\redd \mach{Y}_j \invredd[\mach{c}]\Lookinv{b'_j}$.
Putting all reductions together, we conclude:
\[
\mM\reddh \mZ\redd \tuple{\mZ.\vec R\repl{R_i}{\Lookup{\mach{X}_i}}_{i\in\cR},\mZ.P,[\Lookup{\mach{Y}_1},\dots,\Lookup{\mach{Y}_n}]} \invredd[\mach{c}]\mN
\]

(Transitivity) By induction hypothesis and confluence (Proposition~\ref{prop:confluence}).

(Symmetry) Straightforward from the induction hypothesis.

$(\Leftarrow)$ By Lemma~\ref{lem:onestepisfine}\ref{lem:onestepisfine2} we get $\mM\equiva\mZ$ and $\mN\equiva \mZ$, so we conclude by symmetry and transitivity.
\end{proof}

\begin{prop}\label{prop:cAisnonextcombal}
$\cA_{\,\sima}$ is a non-extensional combinatory algebra.
\end{prop}

\begin{proof} From the calculations in Example~\ref{ex:calculs}, it follows that $\App{\App{\Lookup{\mach{K}}}{a}}{b}\sima a$ and
$\App{\App{\App{\Lookup{\mach{S}}}{a}}{b}}{c}\sima \App{(\App{a}{c})}{(\App{b}{c})}$ hold, for all $a,b,c\in\Addrs$.
Notice that both \am s $\mK$ and $\mS$ are stuck, and $\mK\neq_\Addrs \mS$ since, e.g., $\mK.r\neq\mS.r$.
By Theorem~\ref{thm:CR}, we get $\Lookup\mK\not\sima\Lookup \mS$, whence $\cA_{\,\sima}$ is a combinatory algebra.

To check that $\cA_{\,\sima}$ is not extensional, it is sufficient to exhibit two elements of $\Addrs$ that are extensionally equal, but distinguished modulo $\sima$.
For instance, take $\App{\Lookup\mK}{a}$ and $\App{\Lookup\mK'}{a}$, where $a\in\Addrs$ is arbitrary and $\mK'$ is a different implementation of the combinator $\comb{K}$, namely:
\[
	\begin{array}{lcl}
	\mK' &=& \tuple{\Null,\Null,\RaS (0, 1); \Call 0,[]}, \\
	\mK &=& \tuple{\Null,\RaS (0,-); \Call 0,[]}.
	\end{array}
\]
For all $a,b\in\Addrs$, easy calculations give $\App{\App{\Lookup\mK'}{a}}{b} \sima a$.
Thus, for all $b\in\Addrs$, we have \[
	\App{\App{\Lookup\mK}{a}}{b} \sima a\sima \App{\App{\Lookup\mK'}{a}}{b},
	\]
	whence the two addresses $\App{\Lookup\mK}{a}$ and $ \App{\Lookup\mK'}{a}$ are extensionally equal elements of $\cA_{\,\sima}$.
However, the corresponding \am s are both stuck and $\appT{\mK}{[a]} \neq_\Addrs \appT{\mK'}{[a]}$, because $1 = (\append{\mK}{a}).r \neq (\append{\mK'}{a}).r = 2$.
Since they cannot have a common reduct, we derive $\appT{\mK}{[a]}\not\equiva\appT{\mK'}{[a]}$ by Theorem~\ref{thm:CR}.
We conclude that $\App{\Lookup\mK}{a} \not\sima \App{\Lookup\mK'}{a}$.
\end{proof}

\begin{lem}\label{lem:notlambdalg}
The combinatory algebra $\cA_{\,\sima}$ is not a \lam-model.
\end{lem}

\begin{proof} We need to find $M,N\in\Lambda$ satisfying $M=_\beta N$, while $\cA_{\,\sima}\not\models M = N$.
Take $M = \lam z.(\lam x.x)z =_{\CL} \comb{S(KI)I}$ and $N=\lam x.x =_{\CL} \comb{I}$ where $\comb{I} = \comb{SKK}$.

Recall that $\mach{I} = \mach{\append{S}{\Lookup K,\Lookup K}}$.
Easy calculations give:
\[
	\begin{array}{lll}
	\mach{\append{S}{\App{\Lookup \mK}{\Lookup \mach{I}},\Lookup \mach{I}}} &=&
	\tuple{\Null,\Null,\Null,\RaS 0;\cdots,[\App{\Lookup \mK}{\Lookup \mach{I}},\Lookup\mach{I}]}\\
	&\redh&\tuple{\App{\Lookup \mK}{\Lookup \mach{I}},\Null,\Null,\RaS 1;\cdots,[\Lookup\mach{I}]}\\
	&\redh&\stuck{\tuple{\App{\Lookup \mK}{\Lookup \mach{I}},\Lookup\mach{I},\Null,\RaS 2;\cdots,[]}}
	\end{array}
\]
Similarly,
\[
	\mach{I} = \append{\mach{S}}{\Lookup\mK,\Lookup\mK} \reddh \stuck{\tuple{\Lookup\mach{K},\Lookup\mach{K},\Null,\RaS 2;\cdots,[]}}.
\]
These two machines are both stuck and different modulo $=_\Addrs$ since, e.g., the contents of their register $R_1$ are $\Lookup\mach{I}$ and $\Lookup\mach{K}$ respectively, and it is easy to check that $\Lookup\mach{I}\not\sima \Lookup\mach{K}$.
By Theorem~\ref{thm:CR}, we conclude that $ \App{\App{\Lookup \mS}{(\App{\Lookup\mK}{\Lookup\mach{I}})}}{\Lookup\mach{I}}\not\sima\Lookup\mach{I}$.
\end{proof}

\section{Lambda Models via Applicative Equivalences}\label{sec:lammod}

In the previous section we have seen that the equivalence $\sima$, thus $\equiva$, is too weak to give rise to a model of \lam-calculus (Lemma~\ref{lem:notlambdalg}).
The main problem is that a \lam-term $\lam x.M$ is represented as an \am{} performing a ``$\ins{Load}$'' (to read $x$ from the tape) before evaluating the \am{} corresponding to $M$. Since nothing is applied, the tape is empty and the machine gets stuck thus preventing the evaluation of the subterm $M$.
In order to construct a \lam-model we introduce the equivalence $\simea$ below.

\begin{defi}
Define the relation $\equivea$ as the least equivalence satisfying:
\begin{gather*}
	\infer[\redwerule]{\mM \equivea \mN}{\mM\reddh \mZ \eqea \mN}
	\\[3pt]
	\infer[\extrule]{\mM \equivea \mN}{
		\mach{M}\reddh\stuck{\mM'}&
		\mN\reddh\stuck{\mN'}&
		\forall a\in\Addrs\,.\, \append{\mM}{a} \equivea \append{\mN}{a}
	}
\end{gather*}
We say that $\mM$ and $\mN$ are \emph{applicatively equivalent} whenever $\mM\equivea \mN$. Recall that $\simea$ and $\eqea$ are defined in terms of $\equivea$ as described in Definition~\ref{def:inducingequivalences}.
Also in this case, it is easy to check that $\eqea\ \subseteq\ \equivea$ holds.
\end{defi}

\begin{rem}\label{rem:aboutordinals} The rule $\extrule$ shares similarities with the $(\omega)$-rule in \lam-calculus~\cite[Def.~4.1.10]{Bare}, although being more restricted as only applicable to \am{} that eventually become stuck. In particular, both rules have countably many premises, therefore a derivation of $\mM\equivea\mN$ is a well-founded $\omega$-branching tree (in particular, the tree is countable and there are no infinite paths).
Techniques for performing induction ``on the length of a derivation'' in this kind of systems are well-established, see e.g.~\cite{BarendregtTh,IntrigilaS06}. More details about the underlying ordinals will be given in Section~\ref{sec:consistency}.
\end{rem}

\begin{exas}\label{ex:moreexamples} Convince yourself of the following facts.
\bsub
\item\label{ex:moreexamples1}
	As seen in the proof of Lemma~\ref{lem:notlambdalg}, $\mach{I}$ and $\append{\mach{S}}{\App{\Lookup\mK}{\Lookup \mach{I}},\Lookup\mach{I}}$ both reduce to stuck machines.
	For all $a\in\Addrs$, we have that \[\append{\mach{I}}{a}\reddh\Lookinv a\invredd[\mach{h}] \append{\mach{S}}{\App{\Lookup\mK}{\Lookup\mach{I}},\Lookup\mach{I},a}.\]
	By \extrule, they are applicatively equivalent.
\item\label{ex:moreexamples2}
	Since indeterminate machines $\mach{x}_k$ are not stuck, $\mach{x}_m\equivea\mach{x}_n$ entails $m=n$.
\item\label{ex:moreexamples3}
	Let \[\mach{1} = \tuple{\Null^2,\Load (0,1);\Apply010;\Call 0,[]}.\] It is easy to check that, for all $a,b\in\Addrs$, we have $\append{\mach{1}}{a,b}\reddh \append{\Lookinv{a}}{b}\invredd[\mach{h}] \append{\mach{I}}{a,b}$. However, since   $\append{\mach{I}}{\Lookup\mach{x}_n}\reddh \mach{x}_n$ and $\lnot\stuck{\mach{x}_n}$, one cannot apply \extrule{}, whence (intuitively) they are not applicatively equivalent: $\mach{I} \not\equivea \mach{1}$.
\esub
\end{exas}

\noindent
Actually the inequalities claimed in examples~\ref{ex:moreexamples2}-\ref{ex:moreexamples3} above, i.e.\ $\mach{x}_m\not\equivea\mach{x}_n$ for $m\neq n$ and $\mach{I}\not\equivea\mach{1}$, are difficult to prove formally (see Lemma~\ref{lem:about:equivo}\ref{lem:about:equivo2}).

\begin{lem} Let $\mM,\mN\in\cM$ and $a,b\in\Addrs$.
\begin{enumerate}[(i)]
\item If $\mM\equivea \mN$ then $\append{\mM}{a}\equivea \append{\mN}{a}$.
\item The following rule is derivable:
\[
	\infer[(\mathrm{cong})]{\append{\mM}{a}\equivea \append{\mN}{b}}{\mM \equivea \mN& a \simea b}
\]
\item Therefore, $\simea$ is a congruence on $\cA = (\Addrs,\cdot,\Lookup\mK,\Lookup\mS)$.
\end{enumerate}
\end{lem}

\begin{proof}
$(i)$ By induction on a proof of $\mM\equivea\mN$. Possible cases are:

Case \redwerule. If $\mM\reddh \mZ \eqea \mN$ then $\append{\mM}{a}\reddh\append{\mZ}{a} \eqea \append{\mN}{a}$, by Lemma~\ref{lem:about_red}\ref{lem:about_red2} and the definition of $\eqea$.

Case $\extrule$. Trivial, as the thesis is a premise of this rule.

(Symmetry) and (Transitivity) follow from the induction hypothesis.

$(ii)$ Assume that $\mM\equivea\mN$ and $a\simea b$. Then, we have:
\[
	\begin{array}{lcll}
	\append{\mM}{a}&\eqea&\append{\mM}{b},&\textrm{by reflexivity and }a\simea b,\\
	&\equivea&\append{\mN}{b},&\textrm{by }(i).
	\end{array}
\]
So we conclude by transitivity.

$(iii)$ By Lemma~\ref{lem:equivalence} $\sima$ is an equivalence, by $(ii)$ a congruence.
\end{proof}

We need to show that the congruence $\simea$ is non-trivial, and that the addresses of $\Lookup\mK,\Lookup\mS$ remain distinguished modulo $\simea$.

\begin{lem}\label{lem:about:equivo}
Let $\mM,\mN\in\cM$.
\begin{enumerate}[(i)]
\item\label{lem:about:equivo1}
	If $\mM\equiva \mN$ then $\mM\equivea\mN$.
\item\label{lem:about:equivo2}
	If $\mM\equivea \mN$ and $\mM\reddh\mach{x}_n$ then $\mN\reddh \mach{x}_n$.
\item\label{lem:about:equivo3}
	Hence, the equivalence relation $\simea$ is non-trivial.
\item\label{lem:about:equivo4}
		In particular, $\Lookup \mach{K} \not\simea \Lookup \mach{S}$.
\end{enumerate}
\end{lem}

\begin{proof}
\ref{lem:about:equivo1} 
	Easy.

\ref{lem:about:equivo2} 
	This proof is the topic of Section~\ref{sec:consistency}.

\ref{lem:about:equivo3} 
By~\ref{lem:about:equivo1}, the relation is non-empty. By~\ref{lem:about:equivo2}, $\mach{x}_i \equivea \mach{x}_j $ if and only if $ i = j$, whence there are infinitely many distinguished equivalence classes.

\ref{lem:about:equivo4} 
From Example~\ref{ex:somemachines}, we get:
\[
	\begin{array}{lcl}
	\append{\mach{K}}{\Lookup\mach{K},\Lookup\mach{K},\Lookup\mach{x}_1}&\reddh&\tuple{\Lookup{\mach{x}_1,\RaS -;\Call 0,[]}};\\
	\append{\mach{S}}{\Lookup\mach{K},\Lookup\mach{K},\Lookup\mach{x}_1}&\reddh&\mach{x}_1.\\
	\end{array}
\]
For these machines to be $\equivea$-equivalent, the former machine should reduce to $\mach{x}_1$, by~\ref{lem:about:equivo2}, which is impossible since $\tuple{\Lookup{\mach{x}_1,\RaS -;\Call 0,[]}}$ is stuck.
\end{proof}

\subsection{Constructing a \lam-model}

We define an interpretation transforming a \lam-term with free variables $x_1,\dots,x_n$ into an \am{} reading the values of $\vec x$ from its tape.
The definition is inspired from the well-known categorical interpretation of \lam-calculus into a reflexive object of a cartesian closed category.
In particular, variables are interpreted as projections.
See, e.g.,~\cite{Koymans82} or~\cite{Selinger02} for more details.

\begin{defi}[Auxiliary interpretation]\label{def:categoricalint}
Let $M\in\Lama$ and $x_1,\dots,x_n$ be such that $\FV{M}\subseteq\vec x$.
Define $\CInt[\vec x]{-} : \Lama\to\cM$ by induction as follows:
\[
	\begin{array}{lcll}
	\CInt[\vec x]{x_i} &=& \mach{Pr}_i^n,\textrm{ where }1\le i \le n;\\[1ex]
	\CInt[\vec x]{\cons{a}} &=& \mach{Cons}_a^n$, for $a\in\Addrs;\\[1ex]
	\CInt[\vec x]{MN} &=& \tuple{\Null^{n},\Lookup\CInt[\vec x]{M},\Lookup\CInt[\vec x]{N},\Null,\ins{Apply}_n,[]};\\[1ex]
	\CInt[\vec x]{\lam y.M} &=& \CInt[\vec x,y]{M},&\textrm{ assuming wlog that }y\notin\vec x;\\
	\end{array}
\]
where
\[
	\begin{array}{lcl}
	\mach{Pr}_i^n &=& \tuple{\Null,(\RaS -)^{i-1};\RaS 0;(\RaS -)^{n-i-1};\Call 0,[]},\\[1ex]
	\mach{Cons}_a^n &=& \tuple{a,(\RaS -)^{n};\Call 0,[]},\\	[1ex]
	\ins{Apply}_n &=& \RaS (0,\dots,n-1);
					\Apply n 0 n;\cdots;\Apply n {n-1} n;\\
					&&\Apply {n+1} 0 {n+1};\cdots;\Apply {n+1} {n-1} {n+1};\\
                     &&\Apply n {n+1} {n+2};\Call {n+2}.\\
	\end{array}
\]
\end{defi}

\begin{rem}\label{rem:aux_int}
Let $n\in\nat$, and $T = [a_1,\dots,a_n]\in\Tapes$. We have:
\begin{enumerate}[(i)]
\item\label{rem:aux_int1}
	$\appT{\mach{Pr}_i^n}{T}\reddh \Lookinv{a_i}$, for all $i\,(1\le i\le n)$;
\item\label{rem:aux_int2}
	$\appT{\mach{Cons}_b^n}{T}\reddh \Lookinv{b}$, for all $b\in\Addrs$;
\item\label{rem:aux_int3}
	$\tuple{\Null^{n},\Lookup \mM,\Lookup \mN,\Null,\ins{Apply}_n,T} \reddh\append{(\appT{\mM}{T})}{\Lookup(\appT{\mN}{T})}$.
\end{enumerate}
\end{rem}

\noindent
From now on, whenever writing $\CInt{M}$, we assume that $\FV{M}\subseteq\vec x$.
The following are basic properties of the interpretation map defined above.

\begin{lem}\label{lemma:aboutcatint}
Let $M\in\Lam(\Addrs)$, $n\in\nat$, $\vec x = x_1,\dots,x_n$ and $\vec a = a_1,\dots,a_n\in\Addrs$.
\bsub
\item\label{lemma:aboutcatint1}
	$\CInt[\vec x]{M} = \tuple{\vec R,\RaS (i_1,\dots,i_n);P,[]}$ for some $\Addrs_\Null$-valued registers $\vec R$, program $P$ and indices $i_j\in\nat$.
\item\label{lemma:aboutcatint2}
	If $m<n$ then $\append{\CInt[\vec x]{M}}{a_1,\dots,a_m}\reddh \stuck{}$.
\item\label{lemma:aboutcatint3}
	For all $b\in\Addrs$, we have $\append{\CInt[y,\vec x]{M}}{b} \equivea \CInt[\vec x]{M\subst{y}{\cons b}}$.
\item\label{lemma:aboutcatint4}
	In particular, if $y\notin\FV{M}$ then $\append{\CInt[y,\vec x]{M}}{b} \equivea \CInt{M}$.
\item\label{lemma:aboutcatint5}
	$\append{\CInt{M}}{\vec a\,} \equivea \append{\CInt[x_{\sigma(1)},\dots,x_{\sigma(n)}]{M}}{a_{\sigma(1)},\dots,a_{\sigma(n)}}$ for all permutations~$\sigma$ of $\set{1,\dots,n}$.
\esub
\end{lem}

\begin{proof}[Proof of Lemma~\ref{lemma:aboutcatint}]
\ref{lemma:aboutcatint1} 
	By a straightforward induction on $M$.

\ref{lemma:aboutcatint2} 
	It follows from~\ref{lemma:aboutcatint1}.

\ref{lemma:aboutcatint3} 
	We proceed by structural induction on $M$.
	 By~\ref{lemma:aboutcatint2}, if $\vec x\neq\emptyset$ then both \am s reduce to stuck ones, so we can test the applicative equivalence by applying an arbitrary $\vec a$ and conclude using \extrule{} $n$-times.

Case $M = \cons c$. Then $c\subst{y}{\cons b} = c$, and we have:
	\[
	\append{\CInt[y,\vec x]{\cons c}}{b,\vec a} =
	\append{\mach{Cons}_c^{n+1}}{b,\vec a}
	\reddh \Lookinv c\invredd[\mach{h}]\append{\mach{Cons}_c^{n}}{\vec a}.
	\]

Case $M = x_i$ for some $i\,(1\le i\le n)$. Then $x_i\subst{y}{\cons b} = x_i$ and
	\[
	\append{\CInt[y,\vec x]{x_i}}{b,\vec a} = \append{\mach{Pr}_{i+1}^{n+1}}{b,\vec a}
	\reddh \Lookinv{a_i}\invredd[\mach{h}] \append{\mach{Pr}_i^n}{\vec a}
	= \append{\CInt{x_i}}{\vec a}.
	\]

Case $M = y$. Then $y\subst{y}{\cons b} = \cons b$ and we have:
	\[
	\append{\CInt[y,\vec x]{y}}{b,\vec a\,} =
	\append{\mach{Pr}_{1}^{n+1}}{b,\vec a\,}\reddh
	\Lookinv{b}
	\invredd[\mach{h}]
	\append{\mach{Cons}_b^{n}}{\vec a\,} =
	\append{\CInt{\cons b}}{\vec a\,}.
	\]

Case $M = PQ$. Then $(PQ)\subst{y}{\cons b} = (P\subst{y}{\cons b})(Q\subst{y}{\cons b})$ and we have:
\[
	\begin{array}{llll}
	\append{\CInt[y,\vec x]{PQ}}{b,\vec a\,}&=&
	\tuple{\Null^{n+1},\Lookup\CInt[y,\vec x]{P},\Lookup\CInt[y,\vec x]{Q},\Null,\ins{Apply}_{n+1},[b,\vec a\,]}&\\
	&\reddh&\append{\CInt[y,\vec x]{P}}{b,\vec a,\Lookup (\append{\CInt[y,\vec x]{Q}}{b,\vec a\,})}\\
	&\equivea&\append{\CInt{P\subst{y}{\cons b}}}{\vec a,\Lookup (\append{\CInt{Q\subst{y}{\cons b}}}{\vec a\,})},\textrm{ by IH,}\\
	&\invredd[\mach c]&\tuple{\Null^{n},\Lookup\CInt{P\subst{y}{\cons b}},\Lookup\CInt{Q\subst{y}{\cons b}},\Null,\ins{Apply}_n,[\vec a]}&\\
	&=&\CInt{(P\subst{y}{\cons b})(Q\subst{y}{\cons b})}\\
	&=& \CInt{(PQ)\subst{y}{\cons b}}\\
	\end{array}
\]

Case $M = \lam z.P$, wlog $z\notin y,\vec x$, so $(\lam z.P)\subst{y}{\cons b} = \lam z.P\subst{y}{\cons b}$.
	By~\ref{lemma:aboutcatint2} both machines reduce to stuck ones.
	So we have to apply an extra $a_{n+1}\in\Addrs$.
\[
	\begin{array}{lcll}
	\append{\CInt[y,\vec x]{\lam z.P}}{b,\vec a,a_{n+1}}&=&\append{\CInt[y,\vec x,z]{P}}{b,\vec a,a_{n+1}}&\\
	&\equivea&\append{\CInt[\vec x,z]{P\subst{x}{\cons b}}}{\vec a,a_{n+1}},&\textrm{by IH,}\\
	&=& \append{\CInt{\lam z.P\subst{y}{\cons b}}}{\vec a,a_{n+1}}
	\end{array}
\]

\ref{lemma:aboutcatint4} By~\ref{lemma:aboutcatint3}. 

\ref{lemma:aboutcatint5} By~\ref{lemma:aboutcatint4}, permuting the substitutions. 
\end{proof}

\begin{defi}\label{def:thelambdamodel}
Let $\cS = (\Addrs/_{\simea},\bullet, \Int{-}{-})$, where
\[
	\begin{array}{lcl}
	[a]_{\simea}\bullet [b]_{\simea}&=&[\App{a}{b}]_{\simea},\\[3pt]
	\Int{M}{\rho} &\simea& \Lookup(\append{\CInt{M}}{\rho(x_1),\dots,\rho(x_n)}).\\
	\end{array}
\]
By Lemma~\ref{lemma:aboutcatint}, the definition of $\Int{M}{\rho}$ is independent from the choice of $\vec x$, as long as $\FV{M}\subseteq\vec x$. This is reminiscent of the standard way for defining a syntactic interpretation from a categorical one. (Again, see Koymans's~\cite{Koymans82}.)
\end{defi}

\begin{thm}\label{thm:Sinsonoextslm}
$\cS$ is a syntactic \lam-model.
\end{thm}

\begin{proof} We need to check that the conditions~\ref{def:syntmod1}--\ref{def:syntmod6} from Definition~\ref{def:syntmod} are satisfied by the interpretation function given in Definition~\ref{def:thelambdamodel}.

Take $\vec x = x_1,\dots,x_n$, and write $\rho(\vec x)$ for $\rho(x_1),\dots,\rho(x_n)$.

\ref{def:syntmod1} $\Int{x_i}{\rho} \simea \Lookup(\append{\mach{Pr}_i^n}{\rho(\vec x)})\simea \rho(x_i)$, by Remark~\ref{rem:aux_int}\ref{rem:aux_int1}. 

\ref{def:syntmod2} $\Int{\cons a}{\rho} \simea \Lookup(\append{\mach{Cons}_a^n}{\rho(\vec x)})\simea a$, by Remark~\ref{rem:aux_int}\ref{rem:aux_int2}. 

\ref{def:syntmod3} In the application case, we have: 
\[
	\begin{array}{llll}
	\Int{PQ}{\rho} &\simea&\Lookup(\append{\CInt{PQ}}{\rho(\vec x)})\\
	&=& \tuple{\Null^{n},\Lookup \CInt[\vec x]{P},\Lookup \CInt[\vec x]{Q},\Null,\ins{Apply}_n,[\rho(\vec x)]},&\textrm{by Def.~\ref{def:categoricalint},}\\
	&\simea&\App{\Lookup(\append{\mM}{\rho(\vec x)})}{\Lookup(\append{\mN}{\rho(\vec x)})},&\textrm{by Rem.~\ref{rem:aux_int}\ref{rem:aux_int3},}\\
	&=&\Int{P}{\rho}\bullet \Int{Q}{\rho}
	\end{array}
\]

\ref{def:syntmod4} In the \lam-abstraction case we have, for all $a\in\Addrs$: 
\[
	\App{\Int{\lam y.P}{\rho}}{a} \simea \append{\CInt{\lam y.P}}{\rho(\vec x),a}
  	 \simea \append{\CInt[\vec x,y]{P}}{\rho(\vec x),a}\\
	 \simea\Int{P}{\rho\repl{y}{a}}.
\]

\ref{def:syntmod5} This follows from Lemma~\ref{lemma:aboutcatint}\ref{lemma:aboutcatint4}. 

\ref{def:syntmod6} By definition $\Int{\lam y.M}{\rho} \simea \Lookup(\append{\CInt[\vec x,y]{M}}{\rho(\vec x)})$ and, by Lemma~\ref{lemma:aboutcatint}\ref{lemma:aboutcatint2}, $\append{\CInt[\vec x,y]{M}}{\rho(\vec x)}$ reduces to a stuck \am. 
Similarly, for $\Int{\lam x.N}{\rho}$. We conclude by applying the rule \extrule.
\end{proof}

\begin{rem}
\bsub
\item For closed \lam-terms $M\in\Lamo$, we have $\Int{M}{} = \CInt[]{M}$.
\item It is easy to check that $\Int{\comb{K}}{}\simea\Lookup\mach{K}$ and
$\Int{\comb{S}}{}\simea\Lookup\mach{S}$.
\item More generally, all \am s behaving as the combinator $\comb{K}$ (resp.\ $\comb{S}$) are equated in the model.
\esub
\end{rem}

\begin{lem}\label{lem:Snotext}
The syntactic \lam-model $\cS$ is not extensional.
\end{lem}

\begin{proof} It is enough to check that $\cS\not\models \comb{1} = \comb{I}$. Now, we have:
\[
	\begin{array}{lll}
	\Int{\,\comb{1}\,}{} &=& \tuple{\Null^2,\Lookup\mach{Pr}^2_1,\Lookup\mach{Pr}^2_2,\Null,\ins{Apply}_2,[]};\\
	\Int{\,\comb{I}\ }{} &=& \tuple{\Null,\Load 0;\Call 0,[]}.\\
	\end{array}
\]
By applying an indeterminate machine $\mach{x}_n$, the former reduces to a stuck machine, while the latter reduces to $\mach{x}_n$. By Lemma~\ref{lem:about:equivo}\ref{lem:about:equivo2}, they must be different modulo $\equivea$.
\end{proof}

A difficult problem that arises naturally is the characterization of the \lam-theory induced by the \lam-model $\cS$ defined above.

\begin{prop}\label{prop:aboutThS}
The \lam-theory $\Th{\cS}$ is neither extensional nor sensible.
\end{prop}

\begin{proof} $\Th{\cS}$ is not extensional by Lemma~\ref{lem:Snotext}. To show that it is not sensible, it is enough to check that $\cS\not\models \lam x.\Om = \Om$. Notice that
\[
	\begin{array}{llll}
	\CInt[]{\Om}&=&\tuple{\Lookup\CInt[]{\comb{\Delta}},\Lookup\CInt[]{\comb{\Delta}},\Null,\Apply 012;\Call 2,[]},\\
	&\redh&\append{\CInt[]{\comb{\Delta}}}{\Lookup\CInt[]{\comb{\Delta}}},&\textrm{where:}\\
	\CInt[]{\comb{\Delta}}&=&\tuple{\Null,\Lookup\mach{Pr}^1_1,\Lookup\mach{Pr}^1_1,\Null,\ins{Apply}_1,[]}.\\
	\end{array}
\]
By induction on a derivation of $\mM\equivea \mN$, one checks that $\mM\equivea \mN$ and $\mM\reddh \append{\mach{D}_1}{\mach{D}_2}$ with $\mach{D}_1\simea \mach{D}_2\simea\CInt[]{\comb{\Delta}}$entails $\mN\reddh\append{\mach{D}'_1}{\Lookup\mach{D}'_2}$ for some $\mach{D}'_1\simea\mach{D}'_2\simea\CInt[]{\comb{\Delta}}$. We conclude because the machine $\CInt[]{\lam x.\Om}$ is stuck.
\end{proof}


\section{Consistency Proof via Ordinal Analysis}\label{sec:consistency}
\renewcommand\hole[1]{\llparenthesis#1\rrparenthesis}
\newcommand{\Count}{\omega_1}
\newcommand{\occ}[1]{\mathrm{occ}_\xi(#1)}
\newcommand{\pesc}[1][\alpha]{\approx_{#1}}
\newcommand{\svirg}[1][\alpha]{\sim_{#1}}
\newcommand{\eq}[1][\alpha]{\equiv_{#1}}
\newcommand{\ured}[1][M]{\redh^{\mach{#1}}}
\newcommand{\uredd}[1][M]{\reddh^{\mach{#1}}}
\newcommand{\XX}{\mathbb{X}}
\newcommand{\X}{\mach{X}}
\newcommand{\BB}{\mathbb{B}}
\newcommand{\cMX}{\cM[\XX]^\xi}
\newcommand{\Lup}{\underline{\#}}
\newcommand{\Luinv}[1]{\underline{\#}^{-1}(#1)}

In this section we adapt Barendregt's proof of consistency of $\blam\omega$ (the least \lam-theory closed under the $(\omega)$-rule) to prove Lemma~\ref{lem:about:equivo}\ref{lem:about:equivo2}, which entails the consistency of our system.
First, we need to introduce in our setting the notion of \emph{context} and \emph{underlined reduction}, that are omnipresent techniques in the area of term rewriting systems.

\subsection{Contexts and Underlined Head Reductions}

In \lam-calculus a context is a \lam-term possibly containing occurrences of an algebraic variable, called \emph{hole}, that can be substituted by any \lam-term possibly with capture of free variables.
We will define a \emph{context-machine} similarly, namely as an \am{} possibly having a ``hole'' denoted by $\xi$. Formally, we introduce a new machine having no registers or program, only an empty tape (therefore distinguished from all machines populating $\cM$):
\[
	\xi = \tuple{[]}
\]
We then extend our formalism to include machines working either directly or indirectly with one, or more, occurrences of $\xi$. We wish to ensure the invariant that a machine $\mM$ with no occurrences of $\xi$ maintain as address $\Lookup\mM$ --- for this reason we need to extend the range of addresses in a conservative way.

Consider a countable set $\BB$ of addresses such that $\Addrs\cap\BB = \emptyset$, and write $\XX = \Addrs\cup\BB$ for the set of \emph{extended addresses}. As usual, we set \[\XX_\Null = \XX\cup\set{\Null}.\]

\begin{defi}\label{def:context-machine}
\begin{enumerate}[(i)]
\item
	An \emph{extended machine $\X$} is either of the form
	\begin{itemize}
	\item $\appT{\xi}{T}$ or
	\item $\tuple{\vec R,P,T}$
	\end{itemize}
	where $\vec R$ are $\XX_\Null$-valued registers, $P$ is a valid program, $T\in\Tapes[\XX]$ is an $\XX$-valued tape. We write $\cMX$ for the set of all extended machines.
\item Fix a bijective map $\Lup : \cMX \to \XX$ satisfying $\Lup(\mM)=\Lookup\mM$ for all \am{} $\mM\in\cM$. Write $\Luinv{\cdot} : \XX\to\cMX$ for its inverse.
\item The \emph{number of occurrences} of $\xi$ in $\X\in\cMX$ (resp.\ $R_i$, resp.\ $T$), written $\occ{\X}\in\nat\cup\set{\infty}$ ($\occ{R_i}$, $\occ{T}\in\nat\cup\set{\infty}$), is defined as follows:
\[
	\begin{array}{lcl}
	\occ{\appT{\xi}{T}} &=& 1 + \occ{T};\\[3pt]
	\occ{\tuple{\vec R,P,T}} &=& \occ{T} + \sum_{i=0}^{r-1}\occ{R_i};\\[3pt]
	\occ{[a_1,\dots,a_n]} &=& \occ{\Luinv{a_1}}+\cdots+\occ{\Luinv{a_n}};\\[3pt]
	\occ{R_i}&=&\begin{cases}
	0,&\textrm{if }R_i = \Null,\\
	\occ{\Luinv{a}},&\textrm{if }R_i = a\in\XX.\\
	\end{cases}
	\end{array}
\]
Notice that $\occ{\mM}\in\nat$ entails that $\occ{\mM.R_i},\occ{\mM.T}\in\nat$.
\end{enumerate}
\end{defi}

\noindent
The number of occurrences of $\xi$ in an extended machine $\X$ has been defined to handle the fact that recursively dereferencing all the addresses contained in an extended \am{} might result in a non-terminating process (see Remark~\ref{rem:forever}).

\begin{exas}\label{ex:weird}
The following are examples of extended machines:
\begin{enumerate}[(i)]
\item $\xi$, with $\occ{\xi} = 1$;
\item $\append{\mK}{\Lup{\xi},\Lup{(\append{\xi}{\Lup\xi})}}$, with $\occ{\append{\mK}{\Lup{\xi},\Lup{(\append{\xi}{\Lup\xi})}}} = 3$;
\item\label{ex:weird3} for all $n\in\nat$, $\X_n = \tuple{\Lup\xi,\varepsilon,[\Lup{\X_{n+1}}]}$. In this case, $\occ{\X_0} = \infty$.
\end{enumerate}
\end{exas}

\noindent
As previously mentioned, a key property of contexts in \lam-calculus is that one can plug a \lam-term into the hole and obtain a regular \lam-term.
Similarly, given $\mM\in\cMX$ and $\X\in\cMX$, we can define the \am{} $\X\hole{\mM}$ obtained from $\X$ by recursively substituting (even in the registers/tapes) each occurrence of $\xi$ by $\mM$. However, this operation is well-defined only when $\occ{\X}$ is finite, so we focus on extended machines enjoying this property.

\begin{defi}\
\begin{enumerate}[(i)]
\item A \emph{context-machine} is any $\C\in\cMX$ satisfying $\occ{\C}\in\nat$.
\item Given a context-machine $\C$ and $\mM\in\cM$, define the \am{} $\C\hole{\mM}$ as follows:
\[
	\C\hole{\mM} =\begin{cases}
	\appT{\mM}{T\hole{\mM}},&\textrm{if }\C=\appT{\xi}{T},\\
	\tuple{\vec R\hole{\mM},P,T\hole{\mM}},&\textrm{if }\C=\tuple{\vec R,P,T};\\
	\end{cases}
\]
where (assuming $a\in\XX,T = [a_1,\dots,a_n]\in\Tapes[\XX]$ with $\occ{\Cons a T}\in\nat$):
\[
	\begin{array}{lcl}
	a\hole{\mM} &=& \Lup(\Luinv{a}\hole{\mM});\\[3pt]
	R_i\hole{\mM}&=& \begin{cases}
	\Null&\textrm{if }R_i = \Null,\\
	a\hole{\mM}&\textrm{if }R_i = a;\\[3pt]
	\end{cases}~\\
	T\hole{\mM} &=& [a_1\hole{\mM},\dots,a_n\hole{\mM}].\\[3pt]
	\end{array}
\]
\end{enumerate}
\end{defi}

In the following, when writing $\C\hole{\mM}$ (resp.\ $a\hole{\mM}$, $R_i\hole{\mM}$, $T\hole{\mM}$) we silently assume that the number of occurrences of $\xi$ in $\C$ (resp.\ $a,R_i,T$) is finite.
Let us introduce a notion of reduction for context-machines that allows to mimic the underlined reduction from~\cite{BarendregtTh}. The idea is to decompose a machine $\mN$ as $\mN = \C\hole{\underline{\mM}}$ where $\C$ is a context-machine and $\mM$ the underlined sub-machine.
It is now possible to reduce $\C$ independently from $\mM$ until either the machine reaches a final-state or $\xi$ reaches the head-position. In the latter case, we substitute the head occurrence of $\xi$ by $\mM$, and continue the computation.

\begin{defi}\label{def:weirdreds}\
\bsub
\item\label{def:weirdreds1}
	The head reduction $\redh$ is generalized to extended machines in the obvious way, using $\Lup(\cdot)$ rather than $\Lookup{(\cdot)}$ to compute the addresses.
In particular, the machine $\appT{\xi}{T}\not\redh$ is in final state, but it is not stuck.
\item\label{def:weirdreds2}
	Given $\mM\in\cM$ and $\C\in\cM^\xi$, the \emph{$\mM$-underlined (head-)reduction} $\ured$ is defined by adding to~\ref{def:weirdreds1} the rule
\[
	\appT{\xi}{T}\ured\appT{\mM}{T}.
\]
\esub
\end{defi}

\begin{exas} Let $\C = \append{\mS}{\Lup \xi,\Lup\xi,\Lup \mach{x}_n}$. Then $\C\hole{\mach{K}} = \append{\mS}{\Lookup \mach{K},\Lookup\mach{K},\Lookup \mach{x}_n}$.
\bsub
\item $\C\reddh \append{\xi}{\Lup\mach{x}_n,\Lup{(\append{\xi}{\Lup\mach{x}_n})}}$.
\item$\C\uredd[K] \append{\xi}{\Lup\mach{x}_n,\Lup{(\append{\xi}{\Lup\mach{x}_n})}}
\ured[K]\append{\mach{K}}{\Lup\mach{x}_n,\Lup{(\append{\xi}{\Lup\mach{x}_n})}}\uredd[K] \mach{x}_n$.
\esub
\end{exas}

\begin{lem}\label{lem:chemmeserve}
	For $\C,\C'\in\cMX$ and $\mM,\mN\in\cM$, the following are equivalent:
	\begin{enumerate}
	\item\label{lem:chemmeserve1} $\C\hole{\mM}\reddh \mN$;
	\item\label{lem:chemmeserve2} $\C\reddh^\mM\C'$ and $\C'\hole{\mM} = \mN$.
	\end{enumerate}
\end{lem}

\begin{proof} (\ref{lem:chemmeserve1} $\Rightarrow$~\ref{lem:chemmeserve2})
By induction on the length $n$ of the reduction $\C\hole{\mM}\reddh\mN$.

Case $n = 0$. Trivial, take $\C'=\C$.

Case $n > 0$. Let $\C\hole{\mM}\redh\mN'\reddh\mN$. Split into cases depending on $\C$.

Subcase $\C = \appT{\xi}{T}$, therefore $\C\hole{\mM} = \appT{\mM}{T\hole{\mM}}\redh\mN'$. There are two possibilities:
\begin{itemize}
\item $\mM$ is stuck and $T\neq[]$, say, $T=[a_0,\dots,a_n]$. In this case $\C\hole{\mM} = \tuple{\vec R,\Load i;P,[]}$ and $\mN' = \tuple{\vec R\repl{R_i}{a_0\hole{\mM}},\Load i;P,[a_1\hole{\mM},\dots,a_n\hole{\mM}]}$.
On the other side, $\C\ured \appT{\mM}{T}\ured \C''$ for
\[
	\C'' = \tuple{\vec R\repl{R_i}{a_0},\Load i;P,[a_1,\dots,a_n]}
\]
satisfying $\C''\hole{\mM} = \mN'\reddh\mN$. We conclude by induction hypothesis.
\item $\mM\redh\mM'$. In this case $\mN' = \appT{\mM'}{T\hole{\mM}}$ and $\C\ured \appT{\mM}{T}\ured \C''$ for $\C'' = \appT{\mM'}{T}$ satisfying $\C''\hole{\mM} = \mN'\reddh\mN$. We conclude by induction hypothesis.

Subcase $\C = \tuple{\vec R,P,T}$. By case analysis on $P$. All cases follow easily from the induction hypothesis.
\end{itemize}

(\ref{lem:chemmeserve2} $\Rightarrow$~\ref{lem:chemmeserve1})
By induction on the length $n$ of the reduction $\C\reddh^\mM\C'$.

Case $n=0$. Trivial, take $\mN=\C\hole{\mM}$.

Case $n>0$, i.e.\ $\C\ured\C''\uredd\C'$, where the latter reduction is shorter.

Proceed by case analysis on the shape of $\C$.

Subcase $\C = \appT{\xi}{T}$ and $\C''=\appT{\mM}{T}$.
Then $\mN = \C''\hole{\mM} = \appT{\mM}{T\hole{\mM}} = \C\hole{\mM}$.
Conclude by induction hypothesis.

Subcase $\C = \tuple{\vec R,P,T}$. By case analysis on $P$. All cases follow easily from the induction hypothesis.
\end{proof}

\subsection{Ordinal analysis}

As mentioned in Remark~\ref{rem:aboutordinals}, a derivation of $\mM\equivea\mN$ has the structure of a well-founded $\omega$-branching tree.
Unfortunately, this makes it difficult to prove even simple properties like Lemma~\ref{lem:about:equivo}\ref{lem:about:equivo2}.
We need a more refined system exposing the underlying ordinal and handling the applications of the (Transitivity) rule separately.

\begin{defi}
\begin{enumerate}[(i)]
\item Let $\Count$ be the set of all countable ordinals.
\item If $\pi$ is a derivation of $\mM\equivea\mN$, we define its \emph{length} $\ell(\pi)\in\omega_1$ in the usual inductive way for the rules \redwerule, (Refl.), (Symm.), (Trans.). Concerning the rule $\extrule$ having countably many premises, we set:
\[
	\ell\left(
	\begin{array}{c}
	\infer{\mM \equivea \mN}{
		\mach{M},\mach{N}\reddh\stuck{}&
		\forall a\in\Addrs\,.\, \infer{\append{\mM}{a} \equivea \append{\mN}{a}}{\pi_a}
	}
	\end{array}
	\right) = \sup_{a\in\Addrs}(\ell(\pi_a)+1)
\]
It is easy to check that, if a derivation $\pi$ has premises $(\pi_i)_{i\in \cI}$ for some countable set $\cI$ then $\ell(\pi) > \ell(\pi_i)$ for every $i\in\cI$.
\item For all $\alpha\in\Count$, define $\eq,\svirg,\pesc\,\subseteq\cM^2$ as the least reflexive and symmetric relations closed under the rules of Figure~\ref{fig:Pesiolino}.
\end{enumerate}
\begin{figure}
\begin{gather*}
\infer[(\approx_0)]{\mM\pesc[0] \mN}{\mM\equiva \mN}
\qquad
\infer[(\subseteq^{\approx}_\alpha)]{\mM\svirg\mN}{\mM\pesc\mN}
\qquad
\infer[(\subseteq^{\sim}_\alpha)]{\mM\eq\mN}{\mM\svirg\mN}\\[3pt]
\infer[(\approx_\alpha)]{\mM\pesc\mN}{\mM,\mN\reddh\stuck{}&\forall a\in\Addrs,\,\exists\gamma < \alpha\,.\,\append{\mM}{a} \eq[\gamma] \append{\mN}{a}}\\[3pt]
\begin{array}{ccc}
	\infer[(R_\alpha^\sim)]{\mM\repl{R_i}{a}\svirg\mM\repl{R_i}{b}}{\Lookinv a\svirg \Lookinv b}
	&\quad&
	\infer[(@_\alpha^\sim)]{\append{\mM}{a}\svirg\append{\mM}{b}}{\Lookinv a\svirg \Lookinv b}\\[3pt]
	\infer[(T_\alpha^\sim)]{\appT{\mM}{T}\svirg\appT{\mN}{T}}{\mM\svirg \mN&T\in\Tapes}
	&&
	\infer[(T_\alpha)]{\appT{\mM}{T}\eq\appT{\mN}{T}}{\mM\eq \mN&T\in\Tapes}\\[3pt]
	\infer[(R_\alpha)]{\mM\repl{R_i}{a}\eq\mM\repl{R_i}{b}}{\Lookinv a\eq \Lookinv b}
	&&
	\infer[(@_\alpha)]{\append{\mM}{a}\eq\append{\mM}{b}}{\Lookinv a\eq \Lookinv b}\\[3pt]
\end{array}~\\
\infer[(\le^\approx_\alpha)]{\mM \pesc \mN}{\mM\pesc[\gamma]\mN&\gamma \le \alpha}
\quad
\infer[(\le^\sim_\alpha)]{\mM \svirg \mN}{\mM\svirg[\gamma]\mN&\gamma \le \alpha}
\quad
\infer[(\le_\alpha)]{\mM \eq \mN}{\mM\eq[\gamma]\mN&\gamma \le \alpha}
\\
\infer[(\mathrm{Tr}_\alpha)]{\mM \eq \mN}{\mM\eq\mZ&\mZ\eq\mN}\\[-5ex]
\end{gather*}
\caption{Rules satisfied by $\pesc$, $\svirg$ and $\eq$, beyond reflexivity and symmetry.}\label{fig:Pesiolino}
\end{figure}
\end{defi}
The intuitive meanings of the relations $\eq,\svirg,\pesc$ are the following:
\begin{itemize}
\item $\mM\eq\mN\iff\mM\equivea\mN$ is derivable using the rule $\extrule$ at most $\alpha$ times;
\item $\mM\svirg\mN\iff\mM\eq\mN$ is derivable without using transitivity;
\item $\mM\pesc\mN\iff\mM\equivea\mN$ in case $\alpha = 0$. Otherwise, if $\alpha>0$ then
\item $\mM\pesc\mN\iff\mM\svirg\mN$ follows directly from the rule $\extrule$.
\end{itemize}

\noindent
More precisely, the rules $(\approx_0)$, $(\subseteq^{\approx}_\alpha)$, $(\subseteq^{\sim}_\alpha)$ express the fact that $\equiva\,\subseteq\,\pesc\,\subseteq\,\svirg\,\subseteq\,\eq$.
The rule $(\approx_\alpha)$ allows to prove $\mM \pesc \mN$, provided that both machines eventually get stuck and that $\appT{\mM}{[a]} \eq[\gamma_a] \appT{\mN}{[a]}$ is provable for every address $a$, using a smaller ordinal $\gamma_a < \alpha$.
The rules $(R_\alpha)$, $(@_\alpha)$ and $(T_\alpha)$ (resp.\ $(R_\alpha^\sim)$, $(@_\alpha^\sim)$ and $(T_\alpha^\sim)$) represent the contextuality of the relation $\eq$ (resp.\ $\svirg$).
The rules $(\le^\approx_\alpha)$, $(\le^\sim_\alpha)$ and $(\le_\alpha)$ specify that incrementing the ordinal (from top to bottom) is always allowed.
Finally, $(\mathrm{Tr}_\alpha)$ gives the transitivity of $\eq$.

The following lemma describes formally the intuitive meaning discussed above.
\begin{lem}\label{lem:relalphaprops}
Let $\mM,\mN\in\cM$
\begin{enumerate}[(i)]
\item\label{lem:relalphaprops1}\
	 $\mM\equivea\mN\iff\exists\alpha\in\Count\,.\,\mM\eq\mN$.
\item\label{lem:relalphaprops2}\
	$\mM\eq[0]\mN\iff\mM\equiv_\Addrs\mN$.
\item\label{lem:relalphaprops3}\
	$\mM\eq \mN\iff\exists n\ge0, \mZ_1,\dots,\mZ_n\in\cM\,.\, \mM\svirg\mZ_1\svirg\cdots\svirg\mZ_n =\mN$.
\item\label{lem:relalphaprops4}~\\[-3ex]
$
	\begin{array}{ll}
		\mM\svirg\mN\iff&\exists \mach{C}\in\cMX,\mach{M}',\mach{N}'\in\cM\,.\,\\
		&\mM=\mach{C}\hole{\mM'}\land\mN = \mach{C}\hole{\mN'} \land \mM'\pesc\mN'.\\
		\end{array}
	$
\item\label{lem:relalphaprops5}~\\[-2.7ex]
$
	\begin{array}{lcl}
		\mM\pesc\mN\land \alpha\neq 0&\iff&\mM,\mN\reddh\stuck{}\ \land\\
		&&\forall a\in\Addrs,\exists\gamma<\alpha\,.\, \append{\mM}{a}\eq[\gamma]\append{\mN}{a}.\\
		\end{array}
	$
\end{enumerate}
\end{lem}

\begin{proof}\ref{lem:relalphaprops1} $(\Leftarrow)$ Easy.

$(\Rightarrow)$ By induction on the length of a derivation of $\mM\equivea\mN$.

Case \redwerule. I.e., there exists $\mZ\in\cM$ such that $\mM\reddh\mZ\eqea\mN$.
By Theorem~\ref{thm:CR}, we have $\mM\equiva\mZ$ whence $\mM\eq[0]\mZ$ by $(\pesc[0])$, which implies $\mM\eq\mZ$ for all $\alpha\in\Count$ using the rule $(\le_\alpha)$. Now, consider the set
\[
	\cR = \set{ i \st \mZ.R_i \neq\Null} = \set{ i \st \mN.R_i \neq\Null}
\]
Note that $\cR= \set{i_1,\dots,i_k}$ for some $k<\mZ.r_0 (=\mN.r_0)$. For every $i\in\cR$, let $\mZ.R_i = a_i$ and $\mN.R_i = a'_i$. Also, let $\mZ.T = [b_1,\dots,b_m]$ and $\mN.T = [b'_1,\dots,b'_m]$. By assumption, $a_i\simea a'_i$ and $b_j\simea b'_j$ for every $i\in\cR$, and $j\,(1\le j\le m)$.
By induction hypothesis, $\Lookinv{a_i} \eq[\gamma_i] \Lookinv{a'_i}$ and $\Lookinv{b_j} \eq[\delta_j] \Lookinv{b'_j}$. Using the rule $(<_\alpha)$, the same holds for $\eq[\alpha]$ setting $\alpha = \sup_{i\in\cR,1\le j\le m} \set{\gamma_i,\delta_j}$. Putting everything together, we obtain:
\[
	\begin{array}{lcll}
	\mM&\eq&\mZ = \tuple{\mZ.\vec R,P,[b_1,\dots,b_m]}\\
	&\eq&\tuple{\mZ.\vec R\repl{R_{i_1}}{a'_{i_1}},P,[b_1,\dots,b_m]},&\textrm{by $(R_\alpha)$,}\\
	&\eq&\cdots&\qquad\vdots\\
	&\eq&\tuple{\mZ.\vec R[R_i:= a'_i]_{i\in\cR},P,[b_1,\dots,b_m]},&\textrm{by $(R_\alpha)$,}\\
	&=&\tuple{\mN.\vec R,P,[b_1,\dots,b_m]},&\textrm{by definition,}\\
	&\eq&\tuple{\mN.\vec R,P,[b'_1,b_2,\dots,b_m]},&\textrm{by $(T_\alpha)$,}\\
	&\eq&\cdots&\qquad\vdots\\
	&\eq&\tuple{\mN.\vec R,P,[b'_1,\dots,b'_m]},&\textrm{by $(T_\alpha)$,}\\
	&=&\mN,&\textrm{by definition.}\\
	\end{array}
\]
We conclude by applying the transitivity rule $(\mathrm{Tr}_\alpha)$ that $\mM \eq \mN$.

Case \extrule. By induction hypothesis, for every $a\in\Addrs$, there exists $\gamma_a\in\Count$ such that $\append{\mM}{a} \eq[\gamma_a]\append{\mN}{a}$.
For $\gamma = \sup_{a\in\Addrs}\gamma_a$, we get $\append{\mM}{a} \eq[\gamma]\append{\mN}{a}$ by $(\le_\alpha)$. By $(\pesc)$ we get $\mM\pesc\mN$ for $\alpha=\gamma+1\in\Count$, conclude by $(\subseteq^\approx_\alpha)$ and $(\subseteq^\sim_\alpha)$.

(Reflexivity), (Symmetry) and (Transitivity) follow from the respective property of $\eq$.

Concerning items~\ref{lem:relalphaprops2}--\ref{lem:relalphaprops5} the implication $(\Leftarrow)$ is trivial. We analyze $(\Rightarrow)$.

\ref{lem:relalphaprops2} 
	By induction on a derivation of $\mM\eq[0]\mN$, using Theorem~\ref{thm:CR}.

\ref{lem:relalphaprops3} 
    By induction on a derivation of $\mM\eq\mN$.

    Case $(\subseteq^\sim_\alpha)$. Trivial.

    Case $(R_\alpha)$. I.e., $\mM = \mZ\repl{R_i}{a}$, $\mN = \mZ\repl{R_i}{b}$ and $\Lookinv{a} \eq \Lookinv{b}$. By induction hypothesis, there exist $c_1,\dots,c_k\in\Addrs$ such that
    \[
    	\Lookinv{a}\svirg\Lookinv{c_1}\svirg\cdots\svirg\Lookinv{c_k}=\Lookinv{b}.
    \]
    The case follows by applying the rule $(R^\sim_\alpha)$.

    Case $(@_\alpha)$. Analogous, by applying $(@^\sim_\alpha)$.

    Case $(T_\alpha)$. Analogous, by applying $(T^\sim_\alpha)$.

    Case $(\mathrm{Tr}_\alpha)$. Straightforward from the IH\@.

    Case $(\le_\alpha)$. By IH and $(\le^\sim_\alpha)$.

	Cases (Reflexivity), (Symmetry). Straightforward from the IH\@.

\ref{lem:relalphaprops4} 
	By induction on a derivation of $\mM\svirg\mN$.

	Case $(\subseteq^\approx_\alpha)$. Take $\C = \xi$.

	Case $(R^\sim_\alpha)$. I.e., $\mM = \mZ\repl{R_i}{a}$, $\mN = \mZ\repl{R_i}{b}$ and $\Lookinv{a} \svirg \Lookinv{b}$. By induction hypothesis, there exist $\C'\in\cMX$ having address $c = \Lup{\C'}\in\XX$, $\mM',\mN'\in\cM$ such that $\C'\hole{\mM'} = \Lookinv{a}$, $\C'\hole{\mN'} = \Lookinv{b}$ and $\mM' \pesc\mN'$. We conclude by taking $\C = \mZ\repl{R_i}{c}$.

	Case $(@^\sim_\alpha)$. Analogous.

	Case $(T^\sim_\alpha)$. Take $\C = \appT{\C'}{T}$, where $\C'$ is obtained from the IH\@.

	Case $(\le^\sim_\alpha)$. It follows from the IH, by applying $(\le^\sim_\alpha)$ and $(\le^\approx_\alpha)$.

	Cases (Reflexivity), (Symmetry). Straightforward from the IH\@.

\ref{lem:relalphaprops5} Immediate. 
\end{proof}

Consider now a scenario where $\C\hole \mM\reddh \C'\hole{\mM}$.
Assuming $\mM\pesc\mN$, one might expect that also $\C\hole \mN\reddh \C'\hole{\mN}$ holds.
In general, this is not the case because $\mM$ and $\mN$ might reach the head position and get control of the computation.
Using the underlined (head-)reduction from Definition~\ref{def:weirdreds}\ref{def:weirdreds2} we can substitute $\mN$ for $\mM$ along the reduction (when it comes in head position) and construct a proof of $\C\hole \mN \eq[\gamma] \C'\hole\mN$ having a lower ordinal $\gamma < \alpha$.

\begin{lem}\label{lem:black_magic}
Let $\alpha > 0$, $\C\in\cMX$, $\mM,\mN\in\cM$ such that $\mM\pesc\mN$.
If $\C\redh^\mM\C'$ and $\C'\hole{\mM}\not\reddh\stuck{}$, then there exists $\gamma < \alpha$ such that $\C\hole \mN \eq[\gamma] \C'\hole\mN$.
\end{lem}

\begin{proof} By cases on the shape of $\C$.

Case $\C = \appT{\xi}{T}$ for some $T\in\Tapes[\XX]$ and $\C' = \appT{\mM}{T}$.
From  $\mM\pesc\mN$ and Lemma~\ref{lem:relalphaprops}\ref{lem:relalphaprops5}, we get that $\mM\reddh{\stuck{\mM'}}$ for some $\mM'\in\cM$.
Since $\C'\hole{\mM} = \appT{\mM}{(T\hole{\mM})}$ cannot reduce to a stuck \am, we must have $T\hole{\mM}\neq[]$.
In other words, $T = [a_0,\dots,a_n]$ for some $n\ge 0$.
Notice that, for all $a_i\in\Tapes[\XX]$, we have $a_i\hole{\mN}\in\Addrs$ (by construction).
By Lemma~\ref{lem:relalphaprops}\ref{lem:relalphaprops5}, there exists $\gamma<\alpha$ such that $\append{\mN}{a_0\hole{\mN}} \eq[\gamma] \append{\mM}{a_0\hole{\mN}}$. By definition:
\[
	\C\hole{\mN} = \appT{\mN}{T\hole{\mN}},\textrm{ and }
	\C'\hole{\mN}=\appT{\mM}{T\hole{\mN}}.
\]
So we construct the proof:
\[
	\infer[(T_\gamma)]{\append{\mN}{a_0\hole{\mN},\dots,a_n\hole{\mN}} \eq[\gamma] \append{\mM}{a_0\hole{\mN},\dots,a_n\hole{\mN}}}{
	\append{\mN}{a_0\hole{\mN}}\eq[\gamma] \append{\mM}{a_0\hole{\mN}}
	}
\]

In all the other cases, $\C\hole{\mN}\to_h\C'\hole{\mN}$, therefore $\C\hole{\mN} \eq[0]\C\hole{\mN}$.
\end{proof}

\begin{cor}\label{cor:black_magic}
Let $n\in\nat$, $\alpha > 0$, $\C\in\cMX$, $\mM,\mN\in\cM$.
If $\C\hole{\mM}\reddh\mach{x}_n$ and $\mM\pesc\mN$ then there exists $\gamma < \alpha$ such that $\C\hole{\mN} \eq[\gamma] \mach{x}_n$.
\end{cor}

\begin{proof} Assume $\C\hole{\mM}\reddh\mach{x}_n$. Equivalently, by Lemma~\ref{lem:chemmeserve}, we have $\C\reddh^\mM\mach{x}_n$.
By definition, there exists $\C_1,\dots,\C_k\in\cMX$ such that
\[
	\C = \C_1\to_h^\mM\cdots\to_h^\mM \C_k = \mach{x}_n
\]
Notice that $\C_i\hole{\mM}\reddh \mach{x}_n$ and, since $\lnot\stuck{\mach{x}_n}$, we have $\C_i\hole{\mN}\not\reddh\stuck{}$.
By Lemma~\ref{lem:black_magic}, there exists $\gamma_1,\dots,\gamma_k<\alpha$ such that $\C_i\hole{\mN}\eq[\gamma_i]\C_{i+1}\hole{\mN}$.
By transitivity $(\mathrm{Tr}_\alpha)$ and $(\le_\alpha)$ we obtain $\mM\eq\mach{x}_n$ for $\alpha = \sup_i{\gamma_i}$.
\end{proof}

\begin{prop} Let $\mM,\mN\in\cM$, $\alpha\in\Count$ and $n\in\nat$.
If $\mM\eq \mN$ and $\mN\reddh\mach{x}_n$ then $\mM\reddh\mach{x}_n$.
\end{prop}

\begin{proof} We proceed by induction on $\alpha$. Since we perform a double induction, the induction hypothesis with respect to this induction is called the $\alpha$-IH ($\alpha$-inductive hypothesis).

Case $\alpha = 0$. By Lemma~\ref{lem:relalphaprops}\ref{lem:relalphaprops2}, we get $\mM\equivea \mN\reddh\mach{x}_n$, so we conclude $\mM\reddh\mach{x}_n$ by confluence (Theorem~\ref{thm:CR}) and $\red[i]$-postponement (Lemma~\ref{lem:standardization}).

Case $\alpha > 0$. By Lemma~\ref{lem:relalphaprops}\ref{lem:relalphaprops3}, there exist $\mZ_1,\dots,\mZ_k$ such that
\begin{equation}\label{eq:svirg}
	\mM \svirg\mZ_1\svirg\cdots\svirg\mZ_k=\mN\reddh\mach{x}_n
\end{equation}
By induction on $k$, we prove that~\eqref{eq:svirg} implies $\mM\reddh\mach{x}_n$.
We call this $k$-IH\@.

Subcase $k =0$. Then $\mM =\mN\reddh\mach{x}_n$ and we are done.

Subcase $k >0$. From the $k$-IH we derive $\mZ_1\reddh\mach{x}_n$.
From $\mM\svirg\mZ_1$ and Lemma~\ref{lem:relalphaprops}\ref{lem:relalphaprops4}, there is a context-machine $\C$ such that $\mM = \C[\mM']$ and $\mZ_1 = \C[\mN']$ with $\mM'\pesc\mN'$ and $\C[\mN']\reddh\mach{x}_n$.
By applying Lemma~\ref{lem:black_magic} we obtain $\C[\mM'] \eq[\gamma] \mach{x}_n$ for some $\gamma<\alpha$.
We conclude by applying the $\alpha$-IH\@.
\end{proof}

From this proposition, Lemma~\ref{lem:about:equivo}\ref{lem:about:equivo2} follows by applying Lemma~\ref{lem:relalphaprops}\ref{lem:relalphaprops1}.

\section{Conclusions and Further Works}\label{sec:conclusions}

In this paper, we have shown that it is possible to obtain a model of the untyped \lam-calculus based on a kind of computational machines that operate exclusively on ``addresses'', without any reference to some basic data-type.
The result only depends on the assumption that every machine has a unique address (and \emph{vice versa} every address identifies a machine) and is completely independent from the specific nature of the addresses themselves.

A natural question that can be raised is whether \am s can be seen as a representation of Combinatory Logic's operational semantics in disguise, since their instructions essentially incorporate the contents of the rewriting rules of the basic combinators. To correct this simplistic point of view, observe that the address table map is an arbitrary bijection, whence there are uncountably many possible choices. In particular, address table maps may have arbitrary computational complexity. On the contrary, the operational semantics is constrained to work with the subterms of the current term, i.e.\ it uses a very ``narrow'' address table map. We plan to investigate in future works what possibilities arise from the extra degree of freedom given by the arbitrary nature of this map.

We would like to explore whether the theory of the \lam-model $\cS$ defined in Section~\ref{sec:consistency} depends on the specific nature of the bijection $\Lookup(\cdot) : \Addrs \to\cM$. As discussed in Remark~\ref{rem:forever}, certain ATMs display some peculiarities, since they may create infinite chains of references morally representing infinitary objects. In fact, given an ATM $\#(-)$ and an  injection $f:\nat\to\Addrs$, a simple application of Hilbert's Hotel allows to define a new ATM $\#'(-)$ where machines $(\mM^f_n)_{n\in\nat}$ satisfying $\mM^f_n = \tuple{f(n),\varepsilon,[\Lookup'(\mM^f_{n+1})]}$ exist.
However, these machines are not \lam-definable, whence they should simply constitute non-definable ``junk'' from the model-theoretic perspective. Therefore, we conjecture that $\Th{\cS}$ is actually independent from the choice of the lookup function $\Lookup(\cdot)$.
In case of a positive answer, it would be interesting to provide a complete characterization of the associated \lam-theory.

In Section~\ref{sec:consistency} we have shown that $\Th{\cS}$ is neither extensional nor sensible. This is due to the fact that we kept our construction tight: at each step --- from applicative structure, to combinatory algebra, and finally to \lam-model --- we added the minimal quotient resolving the issue. In order to obtain an extensional model, it would be sufficient to replace the rule $\extrule$ with a form of extensionality non-restricted to machines that become stuck once executed. Similarly, a sensible model can be obtained by collapsing all the addresses of those machines exhibiting a non-terminating behaviour when executed on a number of indeterminates large enough.
These quotients are not difficult to define, but the non-trivial problem becomes to prove that the resulting \lam-model is non-trivial. This is left for further works.

A different line of research, more in the direction of functional programming, is to expand the computational capabilities of \am s by adding simple data-types and the associated basic operations. In fact, although data-types are unnecessary to achieve Turing-completeness, they are desirable to perform arithmetical operations and conditionals.
Preliminary investigations~\cite{IntrigilaMM21} show that extending \am s with numerals, conditional branching, natural numbers basic arithmetic instructions opens the way for representing Plotkin's {\tt PCF}~\cite{Plotkin77}. These investigations show the precise simulation existing between \am 's head reduction and the corresponding evaluation strategy defined on {\tt PCF} extended with explicit substitutions~\cite{LevyM99}. 
We will check if results of this kind extend to the call-by-value untyped setting. To begin with, we plan to study whether \am s can be used to represent the crumbling abstract machines from~\cite{AccattoliCGC19}.


To perform some tests on \am s, we have implemented the formalism both in functional and imperative style. Even if the sources remain for internal use only, some technical choices deserve a discussion. Although not explicitly required by the  definition, any implementation must rely on a computable association between \am s and the corresponding addresses. To implement such a bijection, one could try to use as addresses the actual pointers to the structures representing the machines, but the referenced data might change without affecting the address.
A naive solution consists in defining an association list $\ell$ of type $\Addrs \times \cM$ and an incremental approach. The list $\ell$ is initialized as the empty-list. When a new machine $\mM$ is created, one checks whether $\mM$ belongs to $\pi_2(\ell)$: in the affirmative case there is nothing to do as the machine is already known; otherwise, a new address $a$ is generated and the pair $(a,\mM)$ is added to the list $\ell$. This guarantees that an address uniquely identifies a machine and that, when an address is used, the corresponding machine has already been introduced. For a more optimized solution one should employ the hash-consing technique, allowing to implement the same concept in a more efficient way.
\smallskip

\paragraph{Acknowledgements} This work is partly supported by ANR Project PPS, ANR-19-CE48-0014. We would like to thank Henk Barendregt for interesting discussions concerning the problem of finding a model of \lam-calculus based on recursive functions, as well as the role of the $(\omega)$-rule. 
We are grateful to the anonymous reviewers for the careful reading and insightful suggestions.
We also thank Nicolas M\"unnik for his comments on the paper.

\bibliographystyle{alphaurl}
\bibliography{biblio}


\end{document}